\documentclass[12pt]{article}

\usepackage{amsmath}
\usepackage{amssymb}
\usepackage{amsthm}
\usepackage{hyperref}
\usepackage[pdftex]{graphicx}
\usepackage[english]{babel}
\usepackage{array}
\usepackage{fancyhdr}
\usepackage{enumerate}

\newtheorem{theorem}{Theorem}[section]
\newtheorem{corollary}[theorem]{Corollary}
\newtheorem{lemma}[theorem]{Lemma}

\newtheorem{assumption}[theorem]{Assumption}

\newtheorem{proposition}[theorem]{Proposition}

\newtheorem{definition}[theorem]{Definition}

\newtheorem{example}[theorem]{Example}

\newtheorem{remark}[theorem]{Remark}

\newtheorem{tab}{Table}

\newcommand{\abs}[1]{\left|{#1}\right|}
\def\EE{{\bf E}}
\newcommand{\eps}{\varepsilon}

\title{Capacity of Multilevel NAND Flash Memory Channels}

\author{
\hspace{-5mm} \begin{tabular}{ccc}
Yonglong Li& Aleksandar Kav\v{c}i\'{c}& Guangyue Han\\
University of Hong Kong&University of Hawaii&University of Hong Kong\\
{\em email:} yonglong@hku.hk& {\em email:} kavcic@hawaii.edu&{\em email:} ghan@hku.hk\\
\end{tabular}}

\date{{\normalsize \today}}

\begin{document}\maketitle\thispagestyle{empty}

\begin{abstract}
In this paper, we initiate a first information-theoretic study on multilevel NAND flash memory channels~\cite{kavcic2014} with intercell interference. More specifically, for a multilevel NAND flash memory channel under mild assumptions, we first prove that such a channel is indecomposable and it features asymptotic equipartition property; we then further prove that stationary processes achieve its information capacity, and consequently, as the order tends to infinity, its Markov capacity converges to its information capacity; eventually, we establish that its operational capacity is equal to its information capacity. Our results suggest that it is highly plausible to apply the ideas and techniques in the computation of the capacity of finite-state channels, which are relatively better explored, to that of the capacity of multilevel NAND flash memory channels.
\end{abstract}


\section{Introduction}

As our world is entering a mobile digital era at a lightening pace, NAND flash memories have been seen in a great variety of real-life applications ranging from portable consumer electronics to personal or even enterprise computing. The insatiable demand of greater affordability from consumers  has been driving the industry and academia to relentlessly make use of aggressive technology scaling and multi-level per cell techniques in the bit-cost reduction process. On the other hand though, as their costs continually reduce, flash memories have been more vulnerable to various device or circuit level noises, such as energy consumption, inter-cell interference and program/erase cycling effects, due to the rapidly growing bit density, and maintaining the overall system reliability and performance has become a major concern.

To combat this increasingly imminent issue, various fault-tolerance techniques such as error correction codes have been employed. Representative work in this direction include BCH codes~\cite{sun2007} and LDPC codes~\cite{courtade2011,dong2011}, rank modulation~\cite{andrew2009} and constrained codes~\cite{siegel2014} and so on. The use of such techniques certainly boosts the overall system performance, however, at the expense of reduced memory storage efficiency. As the level of sophistication of such performance boosting techniques drastically escalates, it is of central importance to know their theoretical limit in terms of achieving the maximal cell storage efficiency.

Recently, there have been a number of attempts in response to such a request; see, e.g., ~\cite{dong2011,dong2012,Cai2013, LiJiang2014, Taranalli2015} and references therein. Particularly, in~\cite{dong2011}, the authors have modelled NAND flash memories as communication channels that can capture the major data distortion noise sources including program/erase cycling effects and inter-cell interference in information-theoretic terms. In this direction, slight yet important modifications to enhance the mathematical tractability of the channel model in~\cite{dong2011} have been made in~\cite{kavcic2014}, where multiple communication channels with input inter-symbol interference that are expected to be more amenable to theoretical analysis were explicitly spelled out. On the other hand, with~\cite {kavcic2014} primarily focusing on the optimal detector design, an information-theoretic analysis of the communication channel capacity, which translates to the theoretical limit of memory cell storage efficiency, is still lacking.

Our primary concern in this paper is essentially the one dimensional causal channel model proposed in~\cite{kavcic2014}, which, mathematically, can be characterized by the following system of equations (for justification of such a mathematical formulation of the channel, see~\cite{kavcic2014}):
\begin{align}
Y_0 & = X_0 + W_0 + U_0, \notag\\
\label{sm-2} Y_{n} & = X_{n}+A_{n}X_{n-1}+B_{n}(Y_{n-1}-E_{n-1})+W_{n}+U_{n}, \quad n \geq 1,
\end{align}
where
\begin{itemize}
\item[(i)] $\{X_i\}$ is the channel input process, taking values from a finite alphabet $\mathcal{X}\stackrel{\triangle}{=}\{v_{0}, v_{1}, \cdots, v_{M-1}\}$, and $\{Y_i\}$ is the channel output process, taking values from $\mathbb{R}$.
\item[(ii)] $\{A_{i}\}$, $\{B_{i}\}$, $\{E_i\}$ and $\{W_{i}\}$ are i.i.d. Gaussian random processes with mean $0$ and variance $\sigma_A^2$, $0<\sigma_B^2<1$, $\sigma_E^2$ and $1$, respectively;
\item[(iii)] $\{U_i\}$ is an i.i.d. random process with the uniform distribution over $(\alpha_1,\alpha_2)$, $\alpha_1, \alpha_2 > 0$;
\item[(iv)] $\{A_{i}\}$, $\{B_{i}\}$, $\{E_i\}$, $\{W_{i}\}$, $\{U_i\}$ and $\{X_{i}\}$ are mutually independent.
\end{itemize}
The major differences between our model and that in~\cite{kavcic2014} are as follows:
\begin{itemize}
\item As in most practical scenarios, our channel model has a ``starting'' time $0$, when the channel is not affected by inter-cell interference;
\item An extra assumption in our channel model is that $\sigma_B^2$ is upper bounded by $1$. As established in Lemma~\ref{alphabound}, such an extra assumption will guarantee the boundedness of the channel output power, and thereby the ``stability'' of the channel.
\end{itemize}
Our ultimate goal is to compute the {\em operational capacity} $C$ of the channel (\ref{sm-2}), which, roughly speaking, is defined as the highest rate at which information can be sent with arbitrarily low probability of error. The presence of input and output memory in the channel, however, makes the problem extremely difficult: computing the capacity of channels with memory is a long open problem in information theory. One of the most effective strategies to attack such a difficult problem is the so-called {\em Markov approximation} scheme, which has been extensively exploited in the past decades for computing the capacity of families of finite-state channels (see~\cite{arnoldsimulationbounds, vontobel,randomapproachhan} and references therein). Roughly speaking, the Markov approximation scheme says that, instead of maximizing the mutual information over general input processes, one can do so over Markovian input processes of order $m$ to obtain the so-called $m$-th order {\em Markov capacity}. The effectiveness of this approach has been justified in~\cite{chensiegel}, where, for a class of finite-state channels, the authors showed that as the order $m$ tends to infinity, the sequence of the Markov capacity will converge to the real capacity of the memory channel. It is plausible that the Markov approximation scheme can be applied to other memory channels as well; as a matter of fact, the main result of the present paper is to confirm this for our channel model.

Recently, much progress has been made in computing the Markov capacity of finite-state channels; in particular, a generalized Blahut-Arimoto algorithm and a randomized algorithm have been respectively proposed in~\cite{vontobel} and~\cite{randomapproachhan}, which, under certain conditions, promise convergence to the the Markov capacity. Though there are numerous issues that need to be addressed to justify the applications of the above-mentioned algorithms to our model, the first and foremost question is whether the Markov capacity converges to the real capacity at all. The affirmative answer given in this work, together with other similarities between the channel models, suggests such a framework ``transplantation'' is indeed plausible.

The recursive nature of our channel permits a reformulation into a channel with ``state'': Given the channel input and output $(x_i, y_i)$ at time $i$, the behavior of our channel in the future does not depend on the channel inputs and outputs before time $i$; put if differently, $(x_i, y_i)$ can be regarded as the state for the channel at time $i+1$. Despite the similarities, such a reformulated channel posed new challenges compared with the well-known finite-state channels: The most serious one is that our channel output alphabet is infinite, and as a consequence, the ``indecomposability'' property of our channel, albeit very similar to that of a finite-state channel, is not uniform over all possible channel states; ripple effects of this issue include a number of technical issues, such as the asymptotic equipartition property and even the existence of some fundamental quantities like mutual information rate and capacity.

Which is the reason that in our treatment, some non-trivial technical issues have to be circumvented: We will prove that our channel is ``indecomposable'' in the sense that the behavior of our channel in the distant future is little affected by the channel state in the earlier stages, and a much finer analysis is needed to deal with the above-mentioned non-uniformity issue. The second issue is that the lack of the stationarity of the output process makes it difficult to establish the asymptotic equipartition property for the output process. For this, we observe that the asymptotic mean stationarity~\cite{gray} of the output process makes it possible to apply tools from ergodic theory to establish the existence of the mutual information rate of our channel and further the asymptotic equipartition property of the output process. Another issue is to mix the ``blocked'' processes to obtain a stationary process achieving the information capacity, for which we find an adaptation of Feinstein's method~\cite{Feinstein} as a solution.

The remainder of this paper is organized as follows. In Section~\ref{indecomposability}, we show that the channel~(\ref{sm-2}) is indecomposable, which, among many other applications, ensures the existence of the information capacity of the channel. In Section~\ref{existence}, we show that, when the input $\{X_n\}$ process is stationary and ergodic, $\{Y_n\}$ and $\{X_n,Y_n\}$ possess the asymptotic equipartition property. In Section~\ref{stationarycapacity}, the information capacity is shown to be equal to the stationary capacity and Markov capacity approaches to the information capacity as the Markov order goes to infinity. Eventually, the operational capacity is shown to be equal to the information capacity.

\section{Indecomposability} \label{indecomposability}

In this section, we will prove that our channel (\ref{sm-2}) is ``indecomposable'' in the sense that, in the distant future, it is little affected by the channel state in the earlier stages. Taking the forms of several inequalities in Lemma~\ref{non-uniform-indecomposability}, the indecompoposability property, among many other applications, will ensure that the information capacity of our channel is well-defined.

To avoid the notational cumbersomeness in the computations, we write
$$
\hat{W}_{i}=X_{i}+A_{i}X_{i-1}+W_{i}-B_{i}E_{i-1}+U_i.
$$
It then follows from a recursive application of (\ref{sm-2}) that
\begin{equation}\label{ycv}
Y_{n}=\hat{W}_{n}+B_{n}Y_{n-1}=\sum_{i=k+2}^{n}\hat{W}_{i}\prod_{j=i+1}^{n}B_{j}+Y_{k+1}\prod_{i=k+2}^{n}B_{i}.
\end{equation}
The following lemma gives an upper bound on the moments of the output of the channel~(\ref{sm-2}).
\begin{lemma}\label{alphabound}
There exist $M_2 > 0$ and $\beta>2$ such that for any $n$ and $x_0^n$,
$$
\EE[|Y_n|^{\beta}|X_0^n=x_0^n]\le M_2,
$$
and consequently,
$$
\EE[|Y_n|^{\beta}]\le M_2.
$$
\end{lemma}

\begin{proof}
In this proof, we will simply replace ``$X_0^n=x_0^n$'' in the conditional part of an expectation by $x_0^n$.

It follows from Minkowski's inequality that for any $p \geq 1$
\begin{align*}
(\EE[|Y_n|^p|x_0^n])^{\frac{1}{p}} &\le (\EE[|\hat{W}_n|^p|x_0^n])^{\frac{1}{p}}+(\EE[|B_n|^p|x_0^n])^{\frac{1}{p}}(\EE[|Y_{n-1}|^p|x_0^n])^{\frac{1}{p}}\\
&\le (\EE[|\hat{W}_n|^p|x_0^n])^{\frac{1}{p}}+(\EE[|B_n|^p])^{\frac{1}{p}}(\EE[|Y_{n-1}|^p|x_0^n])^{\frac{1}{p}},
\end{align*}
where we have used the independence between $B_n$ and $Y_{n-1}$, and the independence between $B_n$ and $X_0^n$. Since $\sigma_B^2<1$, there exists $\beta\in (2,3)$ such that $\EE[|B_n|^{\beta}]< 1$. Let
\begin{equation} \label{M_0}
\rho=\EE[|B_n|^{\beta}]^{\frac{1}{\beta}}, \quad M_0=\max\{|\alpha_1|,|\alpha_{2}|,|v_{i}|, \quad i=0,\cdots,m-1\}.
\end{equation}
Then, from Minkowski's inequality and Assumptions (i)-(iv), it follows that
\begin{eqnarray*}
 (\EE[|\hat{W}_n|^{\beta}|x_0^n])^{\frac{1}{\beta}}&\le &(|{x}_n|^{\beta}])^{\frac{1}{\beta}}+(\EE[|{A}_n x_{n-1}|^{\beta}])^{\frac{1}{\beta}}+(\EE[|{W}_n|^{\beta}])^{\frac{1}{\beta}}+(\EE[|B_nE_{n}|^{\beta}])^{\frac{1}{\beta}}+(\EE[|{U}_n|^{\beta}])^{\frac{1}{\beta}}\\
 &\le &2M_0+M_0 (\EE[|{A}_n|^{\beta}])^{\frac{1}{\beta}}+(\EE[|{W}_n|^{\beta}])^{\frac{1}{\beta}}+(\EE[|B_n|^{\beta}])^{\frac{1}{\beta}}(\EE[|E_{n}|^{\beta}])^{\frac{1}{\beta}}\\
 &\stackrel{(a)}{\le} &2M_0+M_0 (\EE[|{A}_n|^{4}])^{\frac{1}{4}}+(\EE[|{W}_n|^{4}])^{\frac{1}{4}}+(\EE[|B_n|^{4}])^{\frac{1}{4}}(\EE[|E_{n}|^{4}])^{\frac{1}{4}},
\end{eqnarray*}
where $(a)$ follows from the inequality $(\EE[|X|^{p}])^{\frac{1}{p}}\le (\EE[|X|^{q}])^{\frac{1}{q}}$ for $0<p<q$. Letting $M_1=2M_0+M_0 (3\sigma_A^{4})^{\frac{1}{4}}+3^{\frac{1}{4}}+(3\sigma_B^4)^{\frac{1}{4}}(3\sigma_E^4)^{\frac{1}{4}}$, we then have
$$
\EE[|\hat{W}_n|^{\beta}|x_0^n]^{\frac{1}{\beta}}\le M_1,
$$
where we have used the fact that the $4$-th moment of a Gaussian random variable with mean $0$ and variance $\sigma^2$ is $3\sigma^4$.

Therefore,
\begin{align}\label{poweriteration}
(\EE[|Y_n|^{\beta]}|x_0^n])^{\frac{1}{\beta}}&\le M_1+\rho(\EE[|Y_{n-1}|^{\beta}|x_0^{n-1}])^{\frac{1}{\beta}},
\end{align}
which implies that
{\begin{eqnarray}
(\EE[|Y_n|^{\beta}|x_0^n])^{\frac{1}{\beta}}&\le &M_1\sum_{i=0}^{n-1}\rho^{i}+\rho^{n}(\EE[|Y_{0}|^{\beta}|x_0])^{\frac{1}{\beta}}\notag\\
&\le& M_1/(1-\rho)+\rho^{n}(\EE[|Y_{0}|^{\beta}|x_0])^{\frac{1}{\beta}}\notag.
\end{eqnarray}}
\noindent It then follows from
$$
\EE[|Y_{0}|^{\beta}|x_0]^{\frac{1}{\beta}}\le (|x_{0}|^{\beta})^{\frac{1}{\beta}}+\EE[|W_{0}|^{\beta}]^{\frac{1}{\beta}}+\EE[|U_{0}|^{\beta}]^{\frac{1}{\beta}}\le 2M_0+\EE[|W_{0}|^{4}]^{\frac{1}{4}},
$$
that there exists $M_2 > 0$ such that for all $x_0^n$,
$$
\EE[|Y_n|^{\beta}|x_0^n]\le M_2,
$$
which immediately implies that
$$
\EE[|Y_n|^{\beta}]\le M_2.
$$
\end{proof}

Lemma~\ref{alphabound} immediately implies the following corollary.
\begin{corollary}\label{uniformintegrable}
$\{Y_n^2\}$ is uniformly integrable and there exists constant $M_3 > 0$ such that
\begin{equation} \label{conditional-power}
\EE[Y_{n}^{2}|X_0^n=x_0^n]\le M_3,
\end{equation}
and consequently,
\begin{equation} \label{power}
\EE[Y_{n}^{2}]\le M_3.
\end{equation}
\end{corollary}
\begin{proof}
The desired uniform integrability immediately follows from Theorem 1.8 in~\cite{liptser} and Lemma~\ref{alphabound}, and the inequality (\ref{conditional-power}) follows from the well-known fact that for any $\beta > 2$,
$$
\EE[Y_{n}^{2}|X_0^n=x_0^n]^{\frac{1}{2}} \le \EE[Y_{n}^{\beta}|X_0^n=x_0^n]^{\frac{1}{\beta}},
$$
which immediately implies (\ref{power}).
\end{proof}

One consequence of Corollary~\ref{uniformintegrable} is the following bounds on the entropy of the channel output.
\begin{corollary}\label{entropybd}
For all $0\le m\le n$,
$$
0 < H(Y_m^{n}) \leq \frac{(n-m+1)\log2\pi e M_3}{2},
$$
where $M_3$ is as in Corollary~\ref{uniformintegrable}.
\end{corollary}

\begin{proof}
For the upper bound, we have
\begin{equation}\label{entropybd1}
 H(Y_m^{n})\le \sum_{i=m}^{n}H(Y_i)\le \frac{(n-m+1)\log2\pi e M_3}{2},
\end{equation}
where~(\ref{entropybd1}) follows from the fact that Gaussian distribution maximizes entropy for a given variance.

For the lower bound, using the chain rule for entropy and the fact that conditioning reduces entropy, we have
\begin{align}
H(Y_m^n) &\geq H(Y_m^{n}|X_m^n) \notag\\
&\ge\sum_{k=m}^{n}H(Y_i|X_{m}^n,Y_{i-1})\notag\\
&\ge \sum_{k=m}^{n}H(Y_i|X_{i-1}^{i},Y_{i-1},E_i,B_i,U_i)\notag\\
&\stackrel{(a)}{=}\sum_{k=m}^{n}H(W_i|X_{i-1}^{i},Y_{i-1},E_i,B_i,U_i) \notag\\
&\stackrel{(b)}{=}\sum_{i=m}^n H(W_i) \notag\\
&=\frac{(n-m+1)\log2\pi e }{2}>0,\notag
 \end{align}
where we have used (\ref{sm-2}) and Assumption (iv) in deriving $(a)$ and $(b)$.
\end{proof}

Fix $k \geq 0$, and for any $x_{k} \in \mathcal{X}$ and $\tilde{y}_{k} \in \mathbb{R}$, define
\begin{align}
 \tilde{Y}_{k+1} &= X_{k+1}+A_{n} x_{k}+B_{n}(\tilde{y}_{k}-E_{k})+W_{k+1}+U_{k+1},\label{newinitial-1}\\
\tilde{Y}_{n} &= X_{n}+A_{n}X_{n-1}+B_{n}(\tilde{Y}_{n-1}-E_{n-1})+W_{n}+U_{n}, \quad n \geq k+1.\label{newinitial-2}
\end{align}
Roughly speaking, $\{\tilde{Y}_n\}$ ``evolves'' in the same way as $\{Y_n\}$, however with different ``conditions'' at time $k$. And similarly as in (\ref{ycv}), we have
\begin{equation} \label{ycv-tilde}
\tilde{Y}_n=\sum_{i=k+2}^{n}\hat{W}_{i}\prod_{j=i+1}^{n}B_{j}+\tilde{Y}_{k+1}\prod_{i=k+1}^{n}B_{i}.
\end{equation}

Below, we will use $f$ (or $p$) with subscripted random variables to denote the corresponding (conditional) probability density function (or mass function). For instance, $f_{Y_{n}|X_k^n,Y_k}(y_{n}|x_k^n, y_k)$ denotes the conditional density of $Y_n$ given $X_k^n=x_k^n$ and $Y_k=y_k$. We may, however, drop the subscripts when there is no confusion and similar notational convention will be followed  throughout the remainder of the paper.

We are now ready for the following lemma that establishes the ``indecomposability'' of our channel. Roughly speaking, the following lemma states that our channel is indecomposable in the sense that the output of our channel in the ``distant future'' is little affected by the ``initial'' inputs and outputs. Compared with the indecomposability property of finite-state channels~\cite{gallagerbook}, our indecomposability does depend on the initial channel inputs and outputs; as a result, a much finer analysis is needed to deal with this non-uniformity issue when one applies Lemma~\ref{non-uniform-indecomposability}.

\begin{lemma}  \label{non-uniform-indecomposability}
a) For any $k \leq n$, $x_{k}^n$, $y_{k}$ and $\tilde{y}_{k}$, we have
\begin{equation}
\int_{-\infty}^{\infty}\left| f_{Y_{n}|X_k^n,Y_k}(y_{n}|x_k^n, y_k)-f_{\tilde{Y}_{n}|X_k^n,\tilde{Y}_k}(y_{n}|x_k^n, \tilde{y}_k)\right| dy_n \le \sigma_{B}^{2(n-k)}(y_{k}^2+\tilde{y}_{k}^2).\notag
\end{equation}

b)For any $k \leq n$, $x_{k}^n$, $y_{k}$ and $\tilde{y}_{k}$, we have
\begin{equation}
\int_{-\infty}^{\infty}y_n^2\left| f_{Y_{n}|X_k^n,Y_k}(y_{n}|x_k^n, y_k)-f_{\tilde{Y}_{n}|X_k^n,\tilde{Y}_k}(y_{n}|x_k^n, \tilde{y}_k)\right| dy_n \le 3\sigma_{B}^{2(n-k)}(y_{k}^2+\tilde{y}_{k}^2).\notag
\end{equation}

c) For any $k, n$, $x_n$ and $y_n$ and $\hat{x}_0^n$, we have
\begin{equation}
\int_{-\infty}^{\infty}\left|f_{Y_n|X_0^n}(\hat{y}|\hat{x}_0^n)-f_{Y_{n+k+1}|X_{n+1}^{n+k+1},X_n,Y_n}(\hat{y}|\hat{x}_0^n,x_n,y_n)\right| d\hat{y} \le  \sigma_{B}^{2n}(\sigma_A^2x_n^{2}+2\sigma_{B}^{2}(y_n^{2}+\sigma_E^2)). \notag
\end{equation}

d) For any $k \leq n$ and any $x_0^n$ with $p_{X_{0}^{n}}(x_{0}^{n})>0$, we have
\begin{equation}
\int_{-\infty}^{\infty} \left| f_{Y_{n}|X_0^n}(y_n|x_0^n)-f_{{Y}_{n}|X_{n-k}^n}(y_n|x_{n-k}^n)\right| dy_n \le \sigma_{B}^{2k}(2\sigma_A^2x_{n-k}^{2}+2\sigma_{B}^{2}(2M_3+2\sigma_E^2)), \notag
\end{equation}
where $M_3$ is as in Corollary~\ref{uniformintegrable}.
\end{lemma}

\begin{proof}
a) Conditioned on $X_{k}^{n}=x_{k}^{n}$, $B_{k+2}^{n}=b_{k+2}^{n}$, $U_{k+1}^{n}=u_{k+1}^{n}$, $E_{k}=e_{k}$, $Y_{k}=y_k,$ and $\tilde{Y}_{k}=\tilde{y}_k$, $Y_{n}$ and $\tilde{Y}_{n}$ are Gaussian random variables with mean $\sum_{i=k+1}^{n}(x_i+u_i)\prod_{j=i+1}^{n}b_{j}$ and respective variances
$$
\sigma^{2}(b_{k+2}^{n},u_{k+1}^{n})=\mbox{Var}(Y_{n}|x_{k}^n,y_k,e_k,b_{k+2}^n,u_{k+1}^{n}), \quad \tilde{\sigma}^{2}(b_{k+2}^{n},u_{k+1}^{n})=\mbox{Var}(\tilde{Y}_{n}|x_{k}^n,\tilde{y}_k,e_k,b_{k+2}^n,u_{k+1}^{n}).
$$
\noindent Note that conditioned on $x_{k}^n,\ b_{k+2}^n,\ u_{k}^{n},\ e_{k},\ y_k$ and $\tilde{y}_k$, $\{\hat{W}_{i}:i=k+2,\cdots,n\}$ and $\{Y_{k+1}, \tilde{Y}_{k+1}\}$ are independent, which implies that
\begin{align*}\sigma^{2}(b_{k+2}^{n},u_{k+1}^{n})=&\mbox{Var}\left(\sum_{i=k+2}^{n}\hat{W}_{i}\prod_{j=i+1}^{n}b_{j}|x_{k}^n,b_{k+2}^n,u_{k+1}^{n}\right)+\mbox{Var}\left(Y_{k+1}\prod_{j=k+2}^{n}b_{j}|x_{k}^n,y_k,b_{k+2}^n,u_{k+1}^{n},e_k\right)
\end{align*}
and
\begin{align*}\tilde{\sigma}^{2}(b_{k+2}^{n},u_{k+1}^{n})=&\mbox{Var}\left(\sum_{i=k+2}^{n}\hat{W}_{i}\prod_{j=i+1}^{n}b_{j}|x_{k}^n,b_{k+2}^n,u_{k+1}^{n}\right)+\mbox{Var}\left(\tilde{Y}_{k+1}\prod_{j=k+2}^{n}b_{j}|x_{k}^n,\tilde{y}_k,b_{k+2}^n,u_{k+1}^{n},e_{k}\right).
\end{align*}
So, we have
\begin{align}\label{im5}
&\hspace{-0.7cm}|\sigma^{2}(b_{k+2}^{n},u_{k+1}^{n})-\tilde{\sigma}^{2}(b_{k+2}^{n},u_{k+1}^{n})|\notag\\
&=\left|\mbox{Var}\left({Y}_{k+1}\prod_{j=k+2}^{n}b_{j}|x_{k}^n,y_k,b_{k+2}^n,u_{k+1}^{n},e_{k}\right)-\mbox{Var}\left(\tilde{Y}_{k+1}\prod_{j=k+2}^{n}b_{j}|x_{k}^n,\tilde{y}_k,b_{k+2}^n,u_{k+1}^{n},e_{k}\right)\right|\notag\\
&=(\left|y_k^2-\tilde{y}_k^2\right|)\sigma_B^{2}\prod_{j=k+2}^{n}b_j^2\notag\\
&\le (y_k^2+\tilde{y}_k^2)\sigma_B^{2}\prod_{j=k+2}^{n}b_j^2.
\end{align}
Now, with the following easily verifiable fact
\begin{equation}
\sigma^{2}(b_{k+2}^{n},u_{k+1}^{n})\ge \mbox{Var}(W_{n})=1\ \mbox{and}\ \tilde{\sigma}^{2}(b_{1}^{n},u_{0}^{n})\ge \mbox{Var}(W_{n})=1\label{im2},
\end{equation}
we conclude that
\begin{align}\label{withoutsqure}
&\hspace{-1cm}\int_{-\infty}^{\infty}\left| f_{Y_{n}|X_k^n,Y_k}(y_{n}|x_k^n, y_k)-f_{\tilde{Y}_{n}|X_k^n,\tilde{Y}_k}(y_{n}|x_k^n, \tilde{y}_k)\right| dy_n \notag \\
&\le \EE\left\{\int_{-\infty}^{\infty}|f(y_{n}|x_{k}^{n},y_k,E_k,B_{k+2}^n,U_{k+1}^{n})-f(y_{n}|x_{k}^{n}, y_{k}, B_{k+2}^{n},U_{k+1}^{n},E_{k})|\,{d}y_{n}\right\}\notag\\
&\stackrel{(a)}{\le} \EE\left\{ \frac{|\sigma^{2}(B_{k+2}^{n},U_{k+1}^{n})-\tilde{\sigma}^{2}(B_{k+2}^{n},U_{k+1}^{n})|\min(\sigma^{2}(B_{k+2}^{n},U_{k+1}^{n}),\tilde{\sigma}^{2}(B_{k+2}^{n},U_{k+1}^{n}))}{\sigma^{2}(B_{k+2}^{n},U_{k+1}^{n})\tilde{\sigma}^{2}(B_{k+2}^{n},U_{k+1}^{n})}\right\} \notag\\
&\stackrel{(b)}{\le} \EE\left\{(y_k^2+\tilde{y}_k^2)\sigma_B^{2}\prod_{j=k+2}^{n}B_j^2\right\} \notag \\
 &=(y_k^2+\tilde{y}_k^2)\sigma_B^{2(n-k)},
\end{align}
where $(a)$ follows from the well-known fact~\cite{madras2010}
$$
\int_{-\infty}^{\infty} \left|\frac{1}{\sqrt{2\pi \sigma_1^2}}e^{-\frac{(x-\mu)^2}{2\sigma_1^{2}}}- \frac{1}{\sqrt{2\pi \sigma_2^2}}e^{-\frac{(x-\mu)^2}{2\sigma_2^{2}}} \right| \,{d} x \le \frac{|\sigma_{1}^{2}-\sigma_{2}^{2}|\min\{\sigma_{1}^{2},\sigma_{2}^{2}\}}{\sigma_{1}^{2}\sigma_{2}^{2}}.
$$
and $(b)$ follows from~(\ref{im5}) and~(\ref{im2}).

b) The proof of b) is similar to a) and the only difference lies in the derivation of (\ref{withoutsqure}), which is given as follows:
\begin{align}
&\hspace{-1cm}\int_{-\infty}^{\infty}y_n^2\left| f_{Y_{n}|X_k^n,Y_k}(y_{n}|x_k^n, y_k)-f_{\tilde{Y}_{n}|X_k^n,\tilde{Y}_k}(y_{n}|x_k^n, \tilde{y}_k)\right| dy_n \notag \\
&\le \EE\left\{\int_{-\infty}^{\infty}y_n^2|f(y_{n}|x_{k}^{n},y_k,E_k,B_{k+2}^n,U_{k+1}^{n})-f(y_{n}|x_{k}^{n}, y_{k}, B_{k+2}^{n},U_{k+1}^{n},E_{k})|\,{d}y_{n}\right\}\notag\\
&\stackrel{(a)}{\le} 3\EE\left\{ |\sigma^{2}(B_{k+2}^{n},U_{k+1}^{n})-\tilde{\sigma}^{2}(B_{k+2}^{n},U_{k+1}^{n})|\right\} \notag\\
&\le 3\EE\left\{(y_k^2+\tilde{y}_k^2)\sigma_B^{2}\prod_{j=k+2}^{n}B_j^2\right\} \notag \\
 &=3(y_k^2+\tilde{y}_k^2)\sigma_B^{2(n-k)},\notag
\end{align}
where $(a)$ follows from the fact that (see Appendix~\ref{Appendix-A} for the proof)
\begin{equation}\label{squaredistance}
\int_{-\infty}^{\infty}x^2 \left|\frac{1}{\sqrt{2\pi \sigma_1^2}}e^{-\frac{(x-\mu)^2}{2\sigma_1^{2}}}- \frac{1}{\sqrt{2\pi \sigma_2^2}}e^{-\frac{(x-\mu)^2}{2\sigma_2^{2}}} \right| \,{d} x \le 3|\sigma_{1}^{2}-\sigma_{2}^{2}|.
\end{equation}

c) This follows from a completely parallel argument as in a).

d) From the assumptions in the channel~(\ref{sm-2}) and Lemma~\ref{uniformintegrable}, it follows that
\begin{align}
&\hspace{-1cm} \int  y_{n-k}^2 f_{{Y}_{n-k}|X_{n-k}^{n}}(y_{n-k}|x_{n-k}^{n}){d}y_{n-k}\\
&=\sum_{\tilde{x}_0^{n-k-1}}\frac{P(X_0^{n-k-1}=\tilde{x}_0^{n-k-1},X_{n-k}^n=x_{n-k}^{n}) \int  y_{n-k}^2 f_{{Y}_{n-k}|X_{0}^{n-k}}(y|\tilde{x}_0^{n-k-1},x_{n-k}^n){d}y_{n-k}}{P(X_{n-k}^{n}=x_{n-k}^{n})}\notag\\
&=\sum_{\tilde{x}_0^{n-k-1}}\frac{P(X_0^{n-k-1}=\tilde{x}_0^{n-k-1},X_{n-k}^n=x_{n-k}^{n}) \int  y_{n-k}^2 f_{{Y}_{n-k}|X_{0}^{n-k}}(y|\tilde{x}_0^{n-k-1},x_{n-k}){d}y_{n-k}}{P(X_{n-k}^{n}=x_{n-k}^{n})}\notag\\
&=\sum_{\tilde{x}_0^{n-k-1}}\frac{P(X_0^{n-k-1}=\tilde{x}_0^{n-k-1},X_{n-k}^n=x_{n-k}^{n})\EE[Y_{n-k}^{2}|\tilde{x}_{0}^{n-k-1},x_{n-k}]}{P(X_{n-k}^{n}=x_{n-k}^{n})} \notag\\
&\le M_3\notag.
\end{align}
We then have
\begin{eqnarray}
&&\hspace{-1.5cm} \int_{-\infty}^{\infty}\left| f_{Y_{n}|X_0^n}(y_n|x_0^n)-f_{{Y}_{n}|X_{n-k}^n}(y_n|x_{n-k}^n)\right|dy_n\notag\\
&=&\int_{-\infty}^{\infty}\left| \int f_{{Y}_{n-k}|X_{0}^{n-k}}(\hat{y}_{n-k}|x_{0}^{n-k})f_{{Y}_{n-k}|X_{n-k}^{n}}(\tilde{y}_{n-k}|x_{n-k}^{n})\right.\notag\\
&{}&\hspace{15mm}\times \left.(f_{Y_{n}|X_{n-k}^n,Y_{n-k}}(y_n|x_{n-k}^n,\hat{y}_{n-k})-f_{{Y}_{n}|X_{n-k}^n,Y_{n-k}}(y_n|x_{n-k}^n,\tilde{y}_{n-k})) {d}\hat{y}_{n-k} d\tilde{y}_{n-k} \right|dy_n\notag\\
&\le&\int f_{{Y}_{n-k}|X_{0}^{n-k}}(\hat{y}_{n-k}|x_{0}^{n-k})f_{{Y}_{n-k}|X_{n-k}^{n}}(\tilde{y}_{n-k}|x_{n-k}^{n})\notag\\
&{}&\hspace{15mm}\times \int_{-\infty}^{\infty} \left|f_{Y_{n}|X_{n-k}^n,Y_{n-k}}(y_n|x_{n-k}^n,\hat{y}_{n-k})-f_{{Y}_{n}|X_{n-k}^n,Y_{n-k}}(y_n|x_{n-k}^n,\tilde{y}_{n-k})\right| dy_n d\hat{y}_{n-k}d\tilde{y}_{n-k}\notag\\
&\stackrel{(a)}{\le}&\int f_{{Y}_{n-k}|X_{0}^{n-k}}(\hat{y}_{n-k}|x_{0}^{n-k})f_{{Y}_{n-k}|X_{n-k}^{n}}(\tilde{y}_{n-k}|x_{n-k}^{n})\sigma_{B}^{2k}(y_{n-k}^2+\tilde{y}_{n-k}^2)d\hat{y}_{n-k} d\tilde{y}_{n-k} \notag\\
&\le&2\sigma_{B}^{2k}M_3,
\end{eqnarray}
where $(a)$ follows from Statement $a)$ in Lemma~\ref{non-uniform-indecomposability}.
\end{proof}

One of the consequences of Lemma~\ref{non-uniform-indecomposability} is the following proposition:
\begin{proposition}\label{pr}
a) Let $X_{n+1}^{2n+1}$ be an independent copy of $X_0^n$. Then for any $k \le n$, any $x \in \mathcal{X}$ and $y \in \mathbb{R}$, we have
\begin{align*}
&\hspace{-1.2cm}|I(X_{0}^{n};Y_{0}^{n})-I(X_{n+1}^{2n+1};Y_{n+1}^{2n+1}|X_n=x,Y_n=y)|\\
&\le 2(k+1)\log M+(n-k)(\sigma_A^2 x^{2}+2\sigma_B^2(y^{2}+\sigma_E^2))\sigma_B^{2k}\log M.
\end{align*}

b) Let $\{X_n\}$ be a stationary process. Then there exist positive constants $M_4,M_5,M_6, M_7,M_8$ and $M_9$ such that for any $m\le k\le n$
\begin{eqnarray}\label{blockdifference}
&&\hspace{-1.2cm} \abs{I(X_0^n; Y_0^n)-I(X_m^{m+n}; Y_m^{m+n})}\notag\\
&\le&\frac{3(k+1)\log 2\pi e M_3}{n+1}+2M_{3}\pi e(M_8+3M_9)\sigma_{B}^{2k}\notag\\
&&+\frac{1}{n+1}(M_4+M_5M_3)+\left(M_6+\frac{4M_1M_3M_7}{(1-\sigma_B)^2}+\frac{12M_3M_7}{(n+1)(1-\sigma_B^2)}\right)\sigma_{B}^{2k}.
\end{eqnarray}
\end{proposition}

\begin{proof}

a) To prove a), we adapt the classical argument in the proof of Theorem $4.6.4$ in~\cite{gallagerbook} as follows.

Using the chain rule for mutual information, we have
\begin{eqnarray*}I(X_0^{n};Y_0^{n})&=&I(X_0^{k};Y_0^{n})+I(X_{k+1}^{n};Y_{k+1}^{n}|X_0^k,Y_0^k)+I(X_{k+1}^{n};Y_{1}^{k}|X_{0}^{k}).
\end{eqnarray*}
It can be verfied that given $X_{0}^{k},$ $X_{k+1}^{n}$ and $Y_0^k$ are independent, which implies that
$$I(X_{k+1}^{n};Y_{0}^{k}|X_{0}^{k})=0.$$
Since $X_{i}$ takes at most $M$ values, we deduce that
$$
|I(X_0^{k};Y_0^{n})|\le (k+1)\log M,
$$
which further implies that
\begin{equation}\label{im17}
I(X_0^{n};Y_0^{n})\le (k+1)\log M+I(X_{k+1}^{n};Y_{k+1}^{n}|X_{0}^{k}, Y_{0}^{k}).
\end{equation}
Similarly, we have, for any $x, y$
\begin{align} \label{im18}
&\hspace{-0.9cm}I(X_{n+1}^{2n+1};Y_{n+1}^{2n+1}|X_n=x,Y_n=y)\\
&\ge -(k+1)\log M +I(X_{n+k+2}^{2n+1};{Y}_{n+k+2}^{2n+1}|X_{n+1}^{n+k+1},Y_{n+1}^{n+k+1},X_n=x,Y_n=y).
\end{align}
It follows from the definition of conditional mutual information that
\begin{align}
I(X_{k+1}^{n};Y_{k+1}^{n}|X_{0}^{k}, Y_{0}^{k})=&\sum_{x_0^k}p(x_0^k)\int f_{Y_0^k|X_0^k}(y_{0}^{k}|x_{0}^{k})I(X_{k+1}^{n};Y_{k+1}^{n}|x_{0}^{k},y_{0}^{k}) {d}y_0^{k}\notag\\
=&\sum_{x_0^k}p(x_0^k)\int f_{Y_k|X_0^k}(y_{k}|x_{0}^{k})I(X_{k+1}^{n};Y_{k+1}^{n}|x_{0}^{k},y_{k}) {d}y_k \label{im16}
\end{align}
and
\begin{align}
&\hspace{-0.8cm} I(X_{n+k+2}^{2n+1};\tilde{Y}_{n+k+2}^{2n+1}|X_{n+1}^{n+k+1},Y_{n+1}^{n+k+1},X_n=x,Y_n=y)\notag\\
&=\sum_{x_0^k} p_{X_{n+1}^{n+k+1}|X_n,Y_n}(x_0^k|x,y)\int \left\{ f_{Y_{n+1}^{n+k+1}|X_{n+1}^{n+k+1},X_{n},Y_{n}}(y_{0}^{k}|x_{0}^{k},x,y)\right.\notag\\
&\hspace{3.2cm}\times \left.I(X_{n+k+2}^{2n+1};Y_{n+k+2}^{2n+1}|X_{n+1}^{n+k+1}=x_0^k,Y_{n+1}^{n+k+1}=y_0^k,X_n=x,Y_n=y)\right\} {d}y_0^{k}\notag\\
&=\sum_{x_0^k} p_{X_{n+1}^{n+k+1}}(x_0^k)\int f_{Y_{n+k+1}|X_{n+1}^{n+k+1},X_{n},Y_{n}}(y_{k}|x_{0}^{k},x,y)I(X_{k+1}^{n};Y_{k+1}^{n}|x_{0}^{k},y_{k})){d}y_{k} \label{im11},
\end{align}
where~(\ref{im11}) follows from
$$
p_{X_{n+1}^{n+k+1}|X_n,Y_n}(\cdot)=p_{X_{n+1}^{n+k+1}}(\cdot)
$$
and
$$
p_{X_{n+k+2}^{2n+1},Y_{n+k+2}^{2n+1}|X_{n+1}^{n+k+1},Y_{n+1}^{n+k+1},X_n,Y_n}(\cdot)=p_{X_{n+k+2}^{2n+1},Y_{n+k+2}^{2n+1}|X_{n+1}^{n+k+1},Y_{n+k+1}}(\cdot)=p_{X_{k+1}^{n},Y_{k+1}^{n}|X_{0}^{k},Y_{k}}(\cdot).
$$
Now, combining~(\ref{im17}),~(\ref{im18}),~(\ref{im16}) and~(\ref{im11}), we conclude that
{\small \begin{align}
&\hspace{-1cm} |I(X_{n+1}^{2n+1};Y_{n+1}^{2n+1}|X_n=x,Y_n=y)-I(X_0^{n};{Y}_0^{n})|\notag\\
\le& 2(k+1)\log M\notag+\sum_{x_0^k}p(x_0^k)\int I(X_{k+1}^{n};Y_{k+1}^{n}|x_{0}^{k},y_{k})\left|f_{Y_k|X_0^k}(y_k|x_0^k)-f_{Y_{n+k+1}|X_{n+1}^{n+k+1},X_n,Y_n}(y_k|x_0^k,x,y)\right| {d}y_{k}\notag\\
\stackrel{(a)}{\le}&2(k+1)\log M+(n-k)\log M\times( \sigma_A^2 x^{2}+2\sigma_B^2(y^{2}+\sigma_E^2))\sigma_B^{2k}\notag,
\end{align}}
where $(a)$ follows from Statement $c)$ in Lemma~\ref{non-uniform-indecomposability} and
$$
I(X_{k+1}^{n};Y_{k+1}^{n}|x_{1}^{k},y_{k})\le H(X_{k+1}^{n})\le (n-k)\log M.
$$

b) To prove b), it suffices to establish that for any $m\le k \leq n$,
\begin{eqnarray}\label{jointentropybd1}
\frac{1}{n+1}|H(Y_0^n)-H(Y_{m}^{m+n})|&\le&\frac{(k+1)\log 2\pi e M_3}{n+1}+\frac{M_4+M_5M_3}{n+1}\notag\\
&&+\left(2M_3M_6+\frac{4M_1M_3M_7}{(1-\sigma_B)^2}+\frac{12M_3M_7}{(n+1)(1-\sigma_B^2)}\right)\sigma_{B}^{2k}
\end{eqnarray}
and
\begin{align}\label{conditionalentropybd1}
\frac{1}{n+1}|H(Y_0^n|X_0^n)-H(Y_{m}^{m+n}|X_{m}^{m+n})|&\le&\frac{2(k+1)\log 2\pi e M_3}{n+1}+2M_{3}\pi e(M_8+3M_9)\sigma_{B}^{2k}.
\end{align}
{\bf Proof of (\ref{jointentropybd1}).} Note that for any $k$,
\begin{eqnarray}
H(Y_{0}^{k-1})&\ge &H(Y_{0}^{k-1}|Y_{k}^{n})\notag\\
&\ge&H(Y_{0}^{k-1}|X_0^n,Y_{k}^{n})\notag\\
&=&H(Y_{0}^{k-1}|X_0^k,Y_{k})\label{markov1}\\
&=&H(Y_0^k|X_0^k)-H(Y_k|X_0^k),\notag
\end{eqnarray}
where~(\ref{markov1}) follows from that $Y_{0}^{k-1}$ is independent of $(X_{k+1}^n,Y_{k+1}^{n})$ given $(X_0^k,Y_{k}).$ Then it follows from Corollary~\ref{entropybd} that
$$\abs{H(Y_{0}^{k-1}|Y_{k}^{n})}\le \frac{(k+1)\log 2\pi e M_3}{2},$$
which further implies that
\begin{eqnarray}
&&\hspace{-1.4cm} \frac{1}{n+1}|H(Y_0^n)-H(Y_{m}^{m+n})|\notag\\
&\le&\frac{1}{n+1}|H(Y_{0}^{k-1}|Y_{k}^{n})-H(Y_{m}^{m+k-1}|Y_{m+k}^{m+n})|+\frac{1}{n+1}|H(Y_k^n)-H(Y_{m+k}^{m+n})|\notag\\
&\le& \frac{(k+1)\log 2\pi e M_3}{n+1}+\frac{1}{n+1}|H(Y_k^n)-H(Y_{m+k}^{m+n})|.\label{jointentropybd}
\end{eqnarray}
Then we have
\begin{eqnarray*}
&&\hspace{-1.4cm}\frac{1}{n+1}|H(Y_k^n)-H(Y_{m+k}^{m+n})|\\
&\le &\frac{1}{n+1}\int |f_{Y_k^n}(y_k^n)\log f_{Y_k^n}(y_k^n)-f_{Y_{m+k}^{m+n}}(y_k^n)\log f_{Y_{m+k}^{m+n}}(y_k^n)| d y_{k}^n\\
&\le &\frac{1}{n+1}D(f_{Y_k^n}(\cdot)||f_{Y_{m+k}^{m+n}}(\cdot))+\frac{1}{n+1}\int \abs{f_{Y_k^n}(y_k^n)-f_{Y_{m+k}^{m+n}}(y_k^n)}\abs{\log f_{Y_{m+k}^{m+n}}(y_k^n)} d y_{k}^n.
\end{eqnarray*}
Using the data processing inequality for relative entropy and the fact (see Appendix~\ref{Appendix-B} for the proof) that there exist positive constants $M_4, M_5$ such that for any $k \leq n$, any $y$ and $x_{k-1}^n$,
\begin{equation}\label{newdensitybd}
|\log f_{Y_{m+k-1}|X_{m+k-1}^{m+n}}(y|x_{k-1}^n)|\le M_4+M_5 y^2,
\end{equation}
we deduce
\begin{eqnarray}
&&\hspace{-1.2cm} \frac{1}{n+1}D(f_{Y_k^n}(\cdot)||f_{Y_{m+k}^{m+n}}(\cdot))\notag\\
&\le&\frac{1}{n+1}D(p_{X_{k-1}^{n}}(\cdot)f_{Y_{k-1}|X_{k-1}^n}(\cdot)||p_{X_{m+k-1}^{m+n}}(\cdot)f_{Y_{m+k-1}|X_{m+k-1}^{m+n}}(\cdot))\notag\\
&=&\frac{1}{n+1}\sum_{x_{k-1}^{n}}p_{X_{k-1}^{n}}(x_{k-1}^{n})\int f_{Y_{k-1}|X_{k-1}^n}(y|x_{k-1}^n)\log\frac{p_{X_{k-1}^{n}}(x_{k-1}^{n})f_{Y_{k-1}|X_{k-1}^n}(y|x_{k-1}^n)}{p_{X_{m+k-1}^{m+n}}(x_{k-1}^{n})f_{Y_{m+k-1}|X_{m+k-1}^{m+n}}(y|x_{k-1}^n)} dy\notag\\
&=&\frac{1}{n+1}\sum_{x_{k-1}^{n}}p_{X_{k-1}^{n}}(x_{k-1}^{n})\int f_{Y_{k-1}|X_{k-1}^n}(y|x_{k-1}^n)\log\frac{f_{Y_{k-1}|X_{k-1}^n}(y|x_{k-1}^n)}{f_{Y_{m+k-1}|X_{m+k-1}^{m+n}}(y|x_{k-1}^n)} dy \notag\\
&\stackrel{(a)}{\le} &\frac{1}{n+1}\sum_{x_{k-1}^{n}}p_{X_{k-1}^{n}}(x_{k-1}^{n})\int f_{Y_{k-1}|X_{k-1}^n}(y|x_{k-1}^n)|\log f_{Y_{m+k-1}|X_{m+k-1}^{m+n}}(y|x_{k-1}^n)| dy \notag\\
&\le&\frac{1}{n+1}\sum_{x_{k-1}^{n}}p_{X_{k-1}^{n}}(x_{k-1}^{n})\int f_{Y_{k-1}|X_{k-1}^n}(y|x_{k-1}^n)(M_4+M_5 y^2)) dy \notag \\
&\le&\frac{1}{n+1}(M_4+M_5M_3)\label{relative2},
\end{eqnarray}
where $(a)$ follows from the fact that $f_{Y_{k-1}|X_{k-1}^n}(y|x_{k-1}^n)\le 1$.
Moreover, from the fact (see Appendix~\ref{Appendix-B} for the proof) that there exist positive constants $M_6, M_7$ such that for any $m \leq n$ and any $y_m^n$,
\begin{equation}\label{densitybd}
\abs{\log f_{Y_{m}^{m+n}}(y_m^n)}\le (n-m+1)M_6+M_7\sum_{i=m}^{n}y_i^2,
\end{equation}
it follows that
\begin{eqnarray}\label{sum}
&&\hspace{-1.4cm}\frac{1}{n+1}\int \abs{f_{Y_k^n}(y_k^n)-f_{Y_{m+k}^{m+n}}(y_k^n)}\abs{\log f_{Y_{m+k}^{m+n}}(y_k^n)} dy_{k}^n \notag\\
&\le&\frac{1}{n+1}\int |f_{Y_k^n}(y_k^n)-f_{Y_{m+k}^{m+n}}(y_k^n)|(nM_6+M_7\sum_{i=k}^{n} y_{i}^{2}) dy_{k}^n\notag\\
&\le&M_6 \int |f_{Y_k^n}(y_k^n)-f_{Y_{m+k}^{m+n}}(y_k^n)|+\frac{M_7}{n+1}\sum_{i=k}^{n}\int dy_{k}^n\, |f_{Y_k^n}(y_k^n)-f_{Y_{m+k}^{m+n}}(y_k^n)| y_{i}^{2} dy_{k}^n.
\end{eqnarray}
Note that
\begin{eqnarray}
&&\hspace{-1.2cm}\int |f_{Y_k^n}(y_k^n)-f_{Y_{m+k}^{m+n}}(y_k^n))| y_{i}^{2} dy_{k}^n\notag\\
 &\le&\sum_{x_0^n}p(x_0^n)\int |f_{Y_k^n|X_0^n}(y_k^n|x_0^n)-f_{Y_{m+k}^{m+n}|X_{m}^{m+n}}(y_k^n|x_0^n)| y_{i}^{2} dy_{k}^n \notag\\
 &=&\sum_{x_0^n}p(x_0^n)\int |f_{Y_k|X_0^k}(y_k|x_0^k)-f_{Y_{m+k}|X_{m}^{m+n}}(y_k|x_0^n)| y_{i}^{2}\prod_{j=k+1}^{n}f_{Y_{j+1}|X_{j-1}^{j},Y_{j-1}}(y_j|x_{j-1}^j,y_{j-1}) dy_{k}^n\notag\\
 &=&\sum_{x_0^n}p(x_0^n)\int |f_{Y_k|X_0^k}(y_k|x_0^k)-f_{Y_{m+k}|X_{m}^{m+n}}(y_k|x_0^n)| y_{i}^{2}\prod_{j=k+1}^{i}f_{Y_{j}|X_{j-1}^{j},Y_{j-1}}(y_j|x_{j-1}^j,y_{j-1}) dy_{k}^i \notag\\
 &=&\sum_{x_0^n}p(x_0^n)\int |f_{Y_k|X_0^k}(y_k|x_0^k)-f_{Y_{m+k}|X_{m}^{m+n}}(y_k|x_0^n)| y_{i}^{2}f_{Y_{i}|X_{k}^{i},Y_{k}}(y_i|x_{k}^i,y_{k}) dy_{k}dy_{i}\notag\\
 &=&\sum_{x_0^n}p(x_0^n)\int |f_{Y_k|X_0^k}(y_k|x_0^k)-f_{Y_{m+k}|X_{m}^{m+n}}(y_k|x_0^n)| \EE[Y_{i}^{2}|X_{k}^{i}=x_{k}^{i},Y_{k}=y_k] dy_{k}\notag\\
 &\stackrel{(a)}{\le}&\sum_{x_0^n}p(x_0^n)\int |f_{Y_k|X_0^k}(y_k|x_0^k)-f_{Y_{m+k}|X_{m}^{m+n}}(y_k|x_0^n)| \left(\frac{2M_1}{(1-\sigma_B)^2}+2\sigma_{B}^{2(i-k)}y_{k}^2)\right) dy_{k} \notag\\
 &=&\sum_{x_0^n}p(x_0^n)\int f_{Y_0|X_0}(y|x_0)f_{Y_{m}|X_{m}^{m+n}}(\tilde{y}|x_{0}^{n}) dyd\tilde{y} \notag\\
 &&\times \int |f_{Y_k|X_0^k,Y_0}(y_k|x_0^k,y)-f_{Y_{m+k}|X_{m}^{m+k},Y_m}(y_k|x_0^k,\tilde{y})| \left(\frac{2M_1}{(1-\sigma_B)^2}+2\sigma_{B}^{2(i-k)}y_{k}^2)\right) dy_{k} \notag\\
 &\stackrel{(b)}{\le} &\sum_{x_0^n}p(x_0^n)\int f_{Y_0|X_0}(y|x_0)f_{Y_{m}|X_{m}^{m+n}}(\tilde{y}|x_{0}^{n})\left(\frac{2M_1}{(1-\sigma_B)^2}(y^2+\tilde{y}^2)\sigma_{B}^{2k}+6(y^{2}+\tilde{y}^2)\sigma_{B}^{2i}\right) dyd\tilde{y} \notag\\
 &=&\left(\frac{2M_1\sigma_{B}^{2k}}{(1-\sigma_B)^2}+6\sigma_{B}^{2i}\right)(\EE [Y_0^{2}]+\EE[ Y_{m}^2])\notag\\
 &\stackrel{(c)}{\le}&\frac{4M_1M_3}{(1-\sigma_B)^2}\sigma_{B}^{2k}+12M_3\sigma_{B}^{2i}\notag,
 \end{eqnarray}
where $(a)$ follows from the same argument in the proof of~(\ref{poweriteration}), $(b)$ follows from Statements $a)$ and $b)$ in Lemma~\ref{non-uniform-indecomposability} and $(c)$ follows from Corollary~\ref{uniformintegrable}. A similar argument can be used to establish that
$$
\int |f_{Y_k^n}(y_k^n)-f_{Y_{m+k}^{m+n}}(y_k^n)| dy_{k}^n \le 2 M_3\sigma_{B}^{2k} ,
$$
which, together with~(\ref{sum}), further implies that
\begin{eqnarray}\label{yKn}
&&\hspace{-2cm}\frac{1}{n+1}\abs{\int  (f_{Y_k^n}(y_k^n)-f_{Y_{m+k}^{m+n}}(y_k^n))\log f_{Y_{m+k}^{m+n}}(y_k^n)} dy_{k}^n\notag\\
&\le&2M_6M_3\sigma_{B}^{2k}+\frac{M_7}{n+1}\sum_{i=k}^{n}\left(\frac{4M_1M_3}{(1-\sigma_B)^2}\sigma_{B}^{2k}+12M_3\sigma_{B}^{2i}\right)\notag\\
&\le&\left(2M_3M_6+\frac{4M_1M_3M_7}{(1-\sigma_B)^2}+\frac{12M_3M_7}{(n+1)(1-\sigma_B^2)}\right)\sigma_{B}^{2k}.
\end{eqnarray}

The desired (\ref{jointentropybd1}) then follows from (\ref{jointentropybd}), (\ref{relative2}) and~(\ref{yKn}).

\noindent{\bf Proof of (\ref{conditionalentropybd1}).} One easily checks that there exist positive constants $M_{8}, M_9$ such that for any $i$, and any $x_{i-1}^i, y$,
\begin{equation}
\EE[Y_i^2|Y_{i-1}=y,X_{i-1}^{i}=x_{i-1}^{i}]\le M_8+M_9 y^2, \notag
\end{equation}
which immediately implies that
\begin{equation}\label{centropy1}
H(Y_i^2|Y_{i-1}=y,X_{i-1}^{i}=x_{i-1}^{i})\le \frac{1}{2}\log 2\pi e(M_8+M_9y^2)\stackrel{(a)}{\le}  \pi e(M_8+M_9y^2),
\end{equation}
where $(a)$ follows from the inequality that $\log x\le x$ for $x>0$.

For any $k\le i<j$, we have
\begin{align}
&\hspace{-0.9cm}|H(Y_{j}|X_{j-1}^{j},Y_{j-1})-H(Y_{i}|X_{i-1}^{i},Y_{i-1})|\notag\\
&=|\sum_{x_0^k}p_{X_{j-k}^j}(x_0^k)\int f_{Y_{j-1}|X_{j-k}^{j}}(y|x_0^k)H(Y_{j}|X_{j-1}^{j}=x_{k-1}^{k},Y_{j-1}=y) dy\notag\\
&\hspace{4mm}-\sum_{x_0^k}p_{X_{i-k}^i}(x_0^k)\int f_{Y_{i-1}|X_{i-k}^{j}}(y|x_0^k)H(Y_{i}|X_{i-1}^{i}=x_{k-1}^{k},Y_{i-1}=y)| dy \notag\\
&\stackrel{(a)}{=}|\sum_{x_0^k}p_{X_{0}^k}(x_0^k)\int H(Y_{2}|X_{1}^{2}=x_{k-1}^{k},Y_{1}=y) (f_{Y_{j-1}|X_{j-k}^{j}}(y|x_0^k)- f_{Y_{i-1}|X_{i-k}^{j}}(y|x_0^k))| dy\notag\\
&\stackrel{(b)}{\le}\sum_{x_0^k}p_{X_{0}^k}(x_0^k)\int \pi e(M_8+M_9y^2) \abs{f_{Y_{j-1}|X_{j-k}^{j}}(y|x_0^k)- f_{Y_{i-1}|X_{i-k}^{j}}(y|x_0^k))} dy\notag\\
&=\sum_{x_0^k}p_{X_{0}^k}(x_0^k)  \int f_{Y_{j-k}|X_{j-k}^{j}}(y_1|x_{0}^{k})f_{Y_{i-k}|X_{i-k}^{i}}(y_2|x_0^{k}) \notag\\
&\hspace{4mm} \times \int |f_{Y_{j-1}|X_{j-k}^{j},Y_{j-k}}(y|x_0^{k-1},y_1)- f_{Y_{i-1}|X_{i-k}^{j},Y_{j-2}}(y|x_0^{k-1},y_2))\pi e(M_8+M_9y^2) dy dy_1dy_2\notag\\
&=\sum_{x_0^k}p_{X_{0}^k}(x_0^k)  \int f_{Y_{j-k}|X_{j-k}^{j}}(y_1|x_{0}^{k})f_{Y_{i-k}|X_{i-k}^{i}}(y_2|x_0^{k})(y_1^2+y_2^2)\sigma_{B}^{2k}\pi e(M_8+3M_9) dy_1dy_2\label{indecomposable2}\\
&\le 2M_{3}\pi e(M_8+3M_9)\sigma_{B}^{2k}\notag,
\end{align}
where $(a)$ follows from the stationarity of $\{X_n\}$ and Assumptions (i),(ii),(iii),(iv) and~(\ref{indecomposable2}) follows from Statements $a)$ and $b)$ in Lemma~\ref{non-uniform-indecomposability} and $(b)$ follows from~(\ref{centropy1}) and
$$
H(Y_{2}|X_{1}^{2}=x_{k-1}^{k},Y_{1}=y)\ge 0.
$$

It then follows that
\begin{eqnarray}
&&\hspace{-1cm}\frac{1}{n+1}|H(Y_0^n|X_0^n)-H(Y_{m}^{m+n}|X_{m}^{m+n})|\notag\\
&\le& \frac{|H(Y_{0}^{k-1}|X_{0}^{k})-H(Y_{m}^{m+k-1}|X_m^{m+k-1})|}{n+1}+\frac{1}{n+1}\sum_{i=k}^{n}|H(Y_{i}|X_{i-1}^{i},Y_{i-1})-H(Y_{i+m}|X_{i+m-1}^{m+i},Y_{m+i-1})|\notag\\
&\le& \frac{2(k+1)\log 2\pi e M_3}{n+1}+\frac{1}{n+1}\sum_{i=k}^{n}|H(Y_{i}|X_{i-1}^{i},Y_{i-1})-H(Y_{i+m}|X_{i+m-1}^{m+i},Y_{m+i-1})|\notag\\
&\le& \frac{2(k+1)\log 2\pi e M_3}{n+1}+\frac{2M_{3}\pi e(M_8+3M_9)(n-k+1)\sigma_{B}^{2k}}{n+1}\notag\\
&\le& \frac{2(k+1)\log 2\pi e M_3}{n+1}+2M_{3}\pi e(M_8+3M_9)\sigma_{B}^{2k},
\end{eqnarray}
as desired
\end{proof}

The information capacity of the channel (\ref{sm-2}) is defined as
\begin{equation}  \label{sm-capacity}
C_{Shannon}=\lim_{n\to \infty}C_{n+1},
\end{equation}
where
$$
C_{n+1}=\frac{1}{n+1}\sup_{p(x_{0}^{n})} I(X_{0}^{n};Y_{0}^{n}).
$$
One consequence of Proposition~\ref{pr} is the existence of the limit in~(\ref{sm-capacity}).
\begin{theorem}
The limit in (\ref{sm-capacity}) exists and therefore $C_{Shannon}$ is well-defined for (\ref{sm-2}).
\end{theorem}

\begin{proof}
Fix $s, t \geq 0$, and let $p^*$ and $q^*$ be input distributions that achieve $C_{s+1}$ and $C_{t+1}$, respectively. From now on, we assume
\begin{equation} \label{p-q}
X_0^{s+t+1} \sim p^*(x_0^s) \times q^*(x_{s+1}^{s+t+1});
\end{equation}
in other words, $X_0^s$ and $X_{s+1}^{s+t+1}$ are independent and distributed according to $p^*$ and $q^*$, respectively. Using (\ref{p-q}) and the assumptions of the channel (\ref{sm-2}), we have
$$
I(X_{s+1}^{s+t+1};Y_{s+1}^{s+t+1}|X_0^s,Y_0^s)=I(X_{s+1}^{s+t+1};Y_{s+1}^{s+t+1}|X_s,Y_s).
$$
Since
$$
I(X_{s+1}^{s+t+1};Y_{s+1}^{s+t+1}|X_s,Y_s)=\sum_{x_s} p_{X_s}(x_s)\int f_{Y_s|X_s}(y_s|x_s) I(X_{s+1}^{s+t+1};Y_{s+1}^{s+t+1}|x_s,y_s) {d}y_s,
$$
we have
\begin{align}
I(X_{s+1}^{s+t+1};Y_{s+1}^{s+t+1}|X_0^s,Y_0^s)&=I(X_{s+1}^{s+t+1};Y_{s+1}^{s+t+1}|X_s,Y_s)\nonumber\\
&=\sum_{x_s}p_{X_s}(x_s)\int f_{Y_s|X_s}(y_s|x_s) I(X_{s+1}^{s+t+1};Y_{s+1}^{s+t+1}|x_s,y_s) {d}y_s \notag\\
&{}\stackrel{(a)}{\ge} \sum_{x_s}p_{X_s}(x_s)\int f_{Y_s|X_s}(y_s|x_s)\left\{ I(X_0^t;Y_0^t)-2(k+1)\log M\right. \nonumber\\
&{}\left.\hspace{4.5mm}-(t-k)\log M\times( \sigma_A^2x_{s}^{2}+2\sigma_B^2(y_{s}^{2}+\sigma_E^2))\sigma_B^{2k}\right\} {d}y_s \nonumber\\
&=I(X_0^t;Y_0^t)-2(k+1)\log M\notag\\
&{}\hspace{4.5mm}-(t-k)\log M\times( \sigma_A^2\EE[X_{s}^{2}]+2\sigma_B^2(\EE[Y_{s}^{2}]+\sigma_E^2))\sigma_B^{2k}\nonumber\\
&\ge I(X_0^t;Y_0^t)-2(k+1)\log M\nonumber\\
&{}\hspace{4.5mm}-(t-k)\log M\times( \sigma_A^2M_0+2\sigma_B^2(M_3+\sigma_E^2))\sigma_B^{2k},
\end{align}
where $(a)$ follows from Statement $a)$ in Proposition~\ref{pr}.
Therefore,
\begin{align}
(s+t+2)C_{s+t+2}&\ge I(X_{0}^{s+t+1};Y_{0}^{s+t+1})\nonumber\\
&\ge I(X_{0}^{s};Y_{0}^{s})+I(X_{s+1}^{s+t+1};Y_{s+1}^{s+t+1}|X_0^s,Y_0^s)\nonumber\\
&\ge (s+1)C_{s+1}+ (t+1)C_{t+1}-2(k+1)\log M\nonumber\\
&{}\hspace{4.5mm}-(t-k)\log M\times( \sigma_A^2M_0+2\sigma_B^2(M_3+\sigma_E^2))\sigma_B^{2k}.
\end{align}
So,
\begin{align}\label{super additive}
C_{s+t+2}&\ge \frac{s+1}{s+1+t+1}C_{s+1}+\frac{t+1}{s+1+t+1}C_{t+1}\nonumber\\
&\hspace{4.5mm}-\frac{2(k+1)\log M}{t+s+2}-\log M\times( \sigma_A^2M_0+2\sigma_B^2(M_3+\sigma_E^2))\sigma_B^{2k}\nonumber.
\end{align}
For any fixed $\eps_{0}>0$, let $k$ be such that $$\log M\times( \sigma_A^2M_0+2\sigma_B^2(M_3+\sigma_E^2))\sigma_B^{2k}\le \frac{\eps_0}{2}$$
and  then let $t> k$ be such that
$$
\frac{2(k+1)\log M}{t+s+2}\le \frac{\eps_0}{2}.
$$
Then for $k$ and $t$ chosen above, we obtain that
$$
(s+t+2)\left\{C_{s+t+2}-\frac{\eps_0}{s+t+2}\right\} \ge(s+1) \left\{C_{s+1}-\frac{\eps_0}{s+1}\right\}+(t+1)\left\{C_{t+1}-\frac{\eps_0}{t+1}\right\}.
$$
By Lemma $2$ on Page $112$ of ~\cite{gallagerbook}, $\lim\limits_{n\to \infty}\left\{C_{n}-\eps_0/n\right\}$ exists and furthermore
$$
\lim_{n\to \infty}C_{n}=\lim_{n\to \infty}\left\{C_{n}-\frac{\eps_0}{n}\right\}=\sup_{n}\left\{C_{n}-\frac{\eps_0}{n}\right\}=\sup_{n \to \infty} C_n.
$$
The proof of the theorem is then complete.
\end{proof}

\section{Asymptotic Mean Stationarity}

One of the main tools that will be used in this work is the so-called asymptotic mean stationarity~\cite{grayprobability}, a natural generalization of stationarity, mostly due to the fact that the output process of our channel is asymptotically mean stationary, rather than stationary. In this section, we give a brief review of notions and results relevant to asymptotic mean stationarity.

Let $\{Y_n\}$ be a real-valued random process over the probability space $(\Omega, \mathcal{F}, P)$. And for $n \in \mathbb{N} \triangleq \{0, 1, 2, \dots\}$, define $\hat{Y}_n: \mathbb{R}^{\mathbb{N}} \to \mathbb{R}$ as the usual coordinate function on $\mathbb{R}^{\mathbb{N}}$ by
$$
\hat{Y}_{n}(x)=x_n \mbox{ for } x=(x_0, x_1, x_2, \dots) \in \mathbb{R}^{\mathbb{N}}.
$$
Let $\mathcal{R}^{\mathbb{N}}$ denote the product Borel $\sigma$-algebra on $\mathbb{R}^{\mathbb{N}}$. By Kolmogorov's extension theorem~\cite{durrettprobability}, there exists an induced probability measure $P_Y$ on $(\mathbb{R}^{\mathbb{N}},\mathcal{R}^{\mathbb{N}})$ such that for any $n \in \mathbb{N}$ and any Borel set $B \subset \mathbb{R}^n$,
$$
P((Y_1,\cdots, Y_n)\in B)=P_Y((\hat{Y}_1,\cdots,\hat{Y}_n)\in B).
$$
So, for ease of presentation only, we sometimes treat the process $\{Y_n\}$ as a function defined as above on the sequence space $\mathbb{R}^{\mathbb{N}}$ equipped with the product Borel $\sigma$-algebra $\mathcal{R}^{\mathbb{N}}$ and the induced measure $P_Y$.

Let $T: \mathbb{R}^{\mathbb{N}} \to \mathbb{R}^{\mathbb{N}}$ be the {\em left shift operator} defined by
$$
Tx=(x_1,x_2,\cdots) \mbox{ for } x=(x_0,x_1,x_2,\cdots) \in \mathbb{R}^{\mathbb{N}}.
$$
A probability measure $\mu$ on $\mathbb{R}^{\mathbb{N}}$ is said to be {\em asymptotically mean stationary} if there exists a probability measure $\bar{\mu}$ such that for any Borel set $A\subset R^{\mathbb{N}}$,
\begin{equation}\label{defams}
\bar{\mu}(A)=\lim_{n\to \infty}\frac{1}{n}\sum_{i=1}^{n}\mu(T^{-i}A);
\end{equation}
And $\bar{\mu}$ in~(\ref{defams}), if it exists, is said to be the {\em stationary mean} of $\mu$. The process $\{Y_n\}$ is said to be {\em asymptotically mean stationary} if the associated measure $P_Y$ is {\em asymptotically mean stationary}.

In the remainder of this paper, we will use subscripted probability measure to emphasize the one with respect to which an expectation is computed; for instance, for a random variable $X$,
$$
\EE_{\mu}(X)=\int X d\mu, \; \mbox{ and } \; H_{\mu}(X) = \int -\log f_X(X) d\mu.
$$

The following theorem gives an analog of Birkhoff's ergodic theorem for asymptotically mean stationary processes.
\begin{theorem}\cite{grayprobability}\label{ergodicthm}
Suppose that $P_Y$ is asymptotically mean stationary with stationary mean $\bar{P}_Y$. If $\EE_{\bar{P}_Y}[|Y_0|] < \infty$, then
$$
\lim_{n\to \infty}\frac{1}{n}\sum_{i=1}^{n} Y_i\quad \mbox{exists} \;\; P_Y-a.s.
$$
\end{theorem}

The following two theorems relate convergences with respect to the measure $P_Y$ and its asymptotic mean $\bar{P}_Y$.
\begin{theorem}\cite{grayprobability}\label{bothergodic}
If $P_Y$ is an asymptotically mean stationary with stationary mean $\bar{P}_Y$, then
$$
\lim_{n\to \infty} \frac{1}{n} \sum_{i=1}^n Y_i \mbox{ exists } P_Y-a.s. \mbox{ if and only if } \quad \lim_{n\to \infty} \frac{1}{n}\sum_{i=1}^n Y_i \mbox{ exists } \bar{P}_Y-a.s.
$$
Also, if the limiting function as above is integrable (with respect to $P_Y$ or $\bar{P}_Y$), then
$$
\EE_{P_Y}\left[\lim_{n\to \infty} \frac{1}{n}\sum_{i=1}^n Y_i\right]=\EE_{\bar{P}_Y}\left[\lim_{n\to \infty} \frac{1}{n}\sum_{i=1}^n Y_i\right].
$$
\end{theorem}

In the following, we will use $\bar{f}_{Y_0^n}(\cdot)$ to denote the density of the probability measures $\bar{P}_Y(Y_0^n\in \cdot)$ with respect to the $(n+1)$-dimensional Lebesgue measure on $\mathbb{R}^{n+1}$.

\begin{theorem}\cite{barron} \label{gsbm}
Suppose that $P_Y$ is asymptotically mean stationary with stationary mean $\bar{P}_Y$, and suppose that for each $n$, there exists $k=k(n)$ such that $I_{P_Y}(Y_1^n; Y_{k+n+1}^{\infty}|Y_{n+1}^{n+k})$ is finite. If for some shift invariant random variable $Z$ (i.e., $Z=Z \circ T$),
$$
\lim_{n\to \infty}\frac{1}{n}\log \bar{f}(X_1^n)= Z, \;\; \bar{P}_Y-a.s.,
$$
then we have
$$
\lim_{n\to \infty}\frac{1}{n}\log {f}(X_1^n)= Z, \;\; P_Y-a.s.
$$
\end{theorem}

\section{Asymptotic Equipartition Property}\label{existence}

Throughout this section, we assume that the input process $\{X_n\}$ is a stationary and ergodic process. As in the previous section, for ease of presentation only, we can assume the process $\{X_n, Y_n\}$ is defined on the sequence space $\mathcal{X}^{\mathbb{N}} \times \mathbb{R}^{\mathbb{N}}$ equipped with the natural product $\sigma$-algebra. Let $P_{XY}$ denote the probability measure on $\mathcal{X}^{\mathbb{N}} \times \mathbb{R}^{\mathbb{N}}$ induced by $\{X_n, Y_n\}$. We will show in this section that $P_{XY}$ is asymptotically mean stationary with stationary mean $\bar{P}_{XY}$, which can be used to establish the asymptotic equipartition property of $\{Y_n\}$ and $\{X_n, Y_n\}$.

For notational simplicity, we often omit the subscripts from the measure associated with a given process when the meaning is clear from the context; e.g., $P_{XY}$ may be simply written as $P$. As opposed to that under the measure $P$, an expectation under $\bar{P}$ will always be emphasized by an extra subscripted $\bar{P}$, i.e., $\EE_{\bar{P}}$. Here, we note that $P$ is the ``original'' meansure, and $\EE_P$ in this section is the same as $\EE$ in other sections.

\begin{theorem}\label{amsproperty}
${P}_Y(\cdot)$ and $P_{XY}(\cdot)$ are asymptotically mean stationary and ergodic.
\end{theorem}
\begin{proof}
{\bf Asymptotic mean stationarity.} We first prove that $P_Y$ is asymptotically mean stationary. To show this, it suffices to show that
\begin{equation}\label{amsexistencef}
\lim_{k\to \infty}P(Y_{k+1}^{k+n}\in A) \mbox{ exists for any $n$ and any Borel set $A \in \mathbb{R}^n$}.
\end{equation}

We will only show (\ref{amsexistencef}) for the case when $n=1$, since the proof for a generic $n$ is rather similar. To this end, consider $|P(Y_{k+1}\in A)-P(Y_{k}\in A)|$. Given $X_{1}^{k+1}=x_0^k,X_0=\tilde{x}_0, Y_0=\tilde{y}_0$, $Y_{k+1}$ is the output of~(\ref{sm-2}) at time $k+1$ starting with
$$Y_1=x_{1}+A_1 \tilde{x}_0+B_{1}(\tilde{y}_0-E_1)+W_1+U_1.$$
Note that
\begin{align*}
P(Y_{k+1}\in A)&=\sum_{\tilde{x}_0,x_{0}^{k}}p_{X_{1}^{k+1}}(x_0^k) p_{X_0|X_{1}^{k+1}}(\tilde{x}_0|x_0^k)\int {f}_{Y_0|X_0}(\tilde{y}|\tilde{x}_0)\, p_{Y_{k+1}|X_{1}^{k+1},X_0,Y_0}( A|x_0^k,\tilde{x}_0, \tilde{y}){d}\tilde{y} .
\end{align*}
Similarly,
\begin{align*}P(Y_{k}\in A)&=\sum_{x_{0}^{k}} p_{X_{0}^{k}}(x_0^k) p_{Y_{k}|X_{0}^{k}}(A|x_0^k)\\
&=\sum_{\tilde{x}_0,x_{0}^{k}}p_{X_{1}^{k+1}}(x_0^k)p_{X_0|X_{1}^{k+1}}(\tilde{x}_0|x_0^k)\int {f}_{Y_0|X_0}(\tilde{y}|\tilde{x}_0)\,p_{{Y}_{k+1}|X_{1}^{k+1}}(A|x_0^k) {d}\tilde{y},
\end{align*}
where $\{Y_n\}$ satisfies~(\ref{sm-2}) with the initial condition $Y_1=X_1+W_1+U_1$.

So, we have
{\small\begin{eqnarray}\label{im10}
&&\hspace{-1.5cm} |P(Y_{k+1}\in A)-P(Y_{k}\in A)|\notag\\
&\le& \sum_{\tilde{x}_0,x_{0}^{k}}p_{X_{1}^{k+1}}(x_0^k)p_{X_0|X_{1}^{k+1}}(\tilde{x}_0|x_0^k)\int {f}_{Y_0|X_0}(\tilde{y}|\tilde{x}_0) |p_{Y_{k+1}|X_{1}^{k+1},X_0,Y_0}(A|x_0^k,\tilde{x}_0, \tilde{y})-p_{Y_{k}|X_{0}^{k}}(A|x_0^k)| d\tilde{y}\nonumber\\
&\le& \sum_{\tilde{x}_0,x_{0}^{k}}p_{X_{1}^{k+1}}(x_0^k)p_{X_0|X_{1}^{k+1}}(\tilde{x}_0|x_0^k)\int {f}_{Y_0|X_0}(\tilde{y}|\tilde{x}_0)\int |{f}_{Y_{k+1}|X_{1}^{k+1},X_0,Y_0}(y|x_0^k,\tilde{x}_0,\tilde{y})-f_{Y_{k+1}|X_{0}^{k}}(y|x_0^k)|{d}y d\tilde{y} \nonumber\\
&\stackrel{(a)}{\leq}&\sum_{\tilde{x}_0,x_{0}^{k}}p_{X_{1}^{k+1}}(x_0^k)p_{X_0|X_{1}^{k+1}}(\tilde{x}_0|x_0^k)\int {f}_{Y_0|X_0}(\tilde{y}|\tilde{x}_0)( \sigma_A^2\tilde{x}_{0}^{2}+2\sigma_B^2(\tilde{y}_{0}^{2}+\sigma_E^2))\sigma_B^{2k} d\tilde{y}\notag\\
&=&( \sigma_A^2\EE[X_{0}^{2}]+2\sigma_B^2(\EE[{Y}_{0}^{2}]+\sigma_E^2))\sigma_B^{2k}\notag\\
&\le& ( \sigma_A^2M_0+2\sigma_B^2(M_3+\sigma_E^2))\sigma_B^{2k},
\end{eqnarray}}
where $(a)$ follows from Statement $c)$ in Lemma~\ref{non-uniform-indecomposability}.

So, the sequence $P(Y_k\in A)$ converges exponentially, which justifies~(\ref{amsexistencef}) for $n=1$. A similar argument can be applied to show that $P_{XY}(\cdot)$ is also asymptotically mean stationary.

\noindent {\bf Ergodicity.} As the ergodicity of $P_{Y}$ follows from that of $P_{XY}$, we only prove the ergodicity of $P_{XY}$. To show the ergodicity, from~\cite{grayprobability}, it suffices to establish that
\begin{eqnarray}\label{ergodicity}
&&\hspace{-1.5cm} \lim_{n\to\infty}\frac{1}{n+1}\sum_{k=0}^{n}P(X_0^{m_1}=x_0^{m_1}, X_{k+1}^{k+m_2}=\hat{x}_{1}^{m_2},Y_{0}^{m_1}\in D,Y_{k+1}^{k+m_2}\in D_2)\nonumber\\
&=&P(X_0^{m_1}=x_0^{m_1}, Y_{0}^{m_1}\in D)\bar{P}(X_{1}^{m_2}=\hat{x}_{1}^{m_2},Y_{1}^{m_2}\in D_2),
\end{eqnarray}
for any $m_1, m_2$, $x_0^{m_1}$ and $\hat{x}_{1}^{m_2}$, any Borel sets $D\subset \mathbb{R}^{m_1+1}$ and $\hat{D}\subset\mathbb{R}^{m_2}$. In the following, we only prove~(\ref{ergodicity}) for $m_1=0$ and $m_2=1$, the proof for general $m_1$ and $m_2$ being similar. Let $\eps$ be an arbitrary positive number. Then we have, for any $\hat{k}$ with $2\sigma_{B}^{2\hat{k}}M_3 \leq \eps$ and sufficiently large $k$,
\begin{eqnarray}
&&\hspace{-1.2cm} P(X_0=x, X_{k+1}=\hat{x},Y_{0}\in D,Y_{k+1}\in \hat{D})\notag\\
&=&\sum_{x_1^k}p_{X_0^{k+1}}(x,x_1^k,\hat{x})\int _{Y_0\in D}f_{Y_0|X_0}(y_0|x)dy_0\int_{y_{k+1}\in \hat{D}}f_{Y_{k+1}|X_0^{k+1},Y_0}(y_{k+1}|x,x_1^k,\hat{x},y_0)dy_{k+1} \notag\\
&\stackrel{(a)}{\ge} &\sum_{x_1^k}p_{X_0^{k+1}}(x,x_1^k,\hat{x})\int _{Y_0\in D}f_{Y_0|X_0}(y_0|x)dy_0\int_{y_{k+1}\in \hat{D}}f_{Y_{k+1}|X_{k-\hat{k}}^{k+1}}(y_{k+1}|x_{k-\hat{k}}^k,\hat{x})dy_{k+1}-2\sigma_{B}^{2\hat{k}}M_3\notag\\
&=&\sum_{\tilde{x}_{1}^{\hat{k}+1}}p_{X_0,X_{k-\hat{k}}^k,X_{k+1}}(x,\tilde{x}_{1}^{\hat{k}+1}, \hat{x})\int _{Y_0\in D}f_{Y_0|X_0}(y_0|x)dy_0\int_{y_{k+1}\in \hat{D}}f_{Y_{k+1}|X_{k-\hat{k}}^{k+1}}(y_{k+1}|\tilde{x}_{1}^{\hat{k}+1},\hat{x})dy_{k+1}\notag\\
&&-2\sigma_{B}^{2\hat{k}}M_3\notag\\
&=&{\small\sum_{\tilde{x}_{1}^{\hat{k}+1}}p_{X_0,X_{k-\hat{k}}^k,X_{k+1}}(x,\tilde{x}_{1}^{\hat{k}+1}, \hat{x})p_{Y_0|X_0}(D|x)P(Y_{k+1} \in \hat{D}|X_{k-\hat{k}}^{k+1}=\tilde{x}_{1}^{\hat{k}+1}\hat{x})}-2\sigma_{B}^{2\hat{k}}M_3\notag\\
&\stackrel{(b)}{\ge}&\sum_{\tilde{x}_{1}^{\hat{k}+1}}p_{X_0,X_{k-\hat{k}}^k,X_{k+1}}(x,\tilde{x}_{1}^{\hat{k}+1}, \hat{x})p_{Y_0|X_0}( D|x)\bar{P}(Y_{k+1}\in \hat{D}|X_{k-\hat{k}}^{k+1}=\tilde{x}_{1}^{\hat{k}+1}\hat{x})-\eps-2\sigma_{B}^{2\hat{k}}M_3\notag\\
&\geq&\sum_{\tilde{x}_{1}^{\hat{k}+1}}p_{X_0,X_{k-\hat{k}}^k,X_{k+1}}(x,\tilde{x}_{1}^{\hat{k}+1}, \hat{x})p_{Y_0|X_0}( D|x)\bar{P}(Y_{k+1}\in \hat{D}|X_{k-\hat{k}}^{k+1}=\tilde{x}_{1}^{\hat{k}+1}\hat{x})-2\eps,\notag
\end{eqnarray}
where $(a)$ follows from Statements $c)$ and $d)$ in Lemma~\ref{non-uniform-indecomposability} and $(b)$ follows from the fact that for sufficiently large $k$,
$$
\left|\bar{P}(Y_{k+1}\in \hat{D}|X_{k-\hat{k}}^{k+1}=\tilde{x}_{1}^{\hat{k}+1}\hat{x})-\bar{P}(Y_{k+1}\in \hat{D}|X_{k-\hat{k}}^{k+1}=\tilde{x}_{1}^{\hat{k}+1}\hat{x})\right|\le \eps.
$$
Then it follows from the ergodicity of $\{X_n\}$ that

\begin{eqnarray}\label{ergodicity2}
&&\hspace{-1.2cm} \lim_{n\to\infty}\frac{1}{n+1}\sum_{k=0}^{n}P(X_0=x, X_{k+1}=\hat{x},Y_{0}\in D,Y_{k+1}\in \hat{D})\nonumber\\
&\ge&\lim_{n\to\infty}\frac{1}{n+1}\sum_{k=0}^{n}\sum_{\tilde{x}_{1}^{\hat{k}+1}}p_{X_0,X_{k-\hat{k}}^{k+1}}(x,\tilde{x}_{1}^{\hat{k}+1},\hat{x})p_{Y_0|X_0}( D|x)\bar{P}(Y_{k+1}\in \hat{D}|X_{k-\hat{k}}^{k+1}=\tilde{x}_{1}^{\hat{k}+1}\hat{x})-2\eps\notag\\
&=&\sum_{\tilde{x}_{1}^{\hat{k}+1}}P(X_0=x)p_{X_{k-\hat{k}}^{k+1}}(\tilde{x}_{1}^{\hat{k}+1}, \hat{x})p_{Y_0|X_0}( D|x)\bar{P}(Y_{k+1}\in \hat{D}|X_{k-\hat{k}}^{k+1}=\tilde{x}_{1}^{\hat{k}+1}\hat{x})-2\eps,\notag\\
&=&P(X_0=x_0,Y_{0}\in D)\bar{P}(X_{k+1}=\hat{x}_{1},Y_{k+1}\in \hat{D})-2\eps.
\end{eqnarray}
Through a parallel argument, we can show that
\begin{eqnarray}\label{ergodicity3}
&&\hspace{-1.5cm} \lim_{n\to\infty}\frac{1}{n+1}\sum_{k=0}^{n}P(X_0=x_0, X_{k+1}=\hat{x}_{1},Y_{0}\in D,Y_{k+1}\in \hat{D})\nonumber\\
&\le &P(X_0=x_0,Y_{0}\in D)\bar{P}(X_{k+1}=\hat{x}_{1},Y_{k+1}\in \hat{D})+2\eps.
\end{eqnarray}
Then the desired result follows from~(\ref{ergodicity2}) and~(\ref{ergodicity3}).
\end{proof}

Using Corollary~\ref{uniformintegrable}, we can prove the following result, which strengthens~(\ref{amsexistencef}) and whose proof can be found in Appendix~\ref{Appendix-B}.

\begin{lemma}\label{amsexistence}
For any fixed $n$,
\begin{equation}\label{im31}
\lim_{k\to \infty}f_{Y_{k}^{k+n}}(\cdot)=\bar{f}_{Y_0^n}(\cdot),
\end{equation}
and furthermore
$$\EE_{\bar{P}}[Y_i^2]=\lim_{n\to \infty}\EE_P[Y_i^2]<\infty.$$
\end{lemma}

Using Theorem~\ref{gsbm}, we can prove the following lemma, which will be used to prove the asymptotic equipartition property for the output $\{Y_n\}$ of the channel~(\ref{sm-2}).
\begin{lemma}\label{existenceentropy}
There exists some constant $a$ such that
$$
\lim_{n\to \infty}\frac{1}{n+1}\log f(Y_0^n)=a, \;\; P-a.s.
$$
\end{lemma}

\begin{proof}
In order to invoke Theorem~\ref{gsbm}, we need to prove that for any $n$, there exists $k(n)$ such that
\begin{equation}\label{finiteness}I(Y_0^n;Y_{n+k(n)+1}^{\infty}|Y_{n+1}^{n+k(n)})< \infty.\end{equation}
and
\begin{equation}\label{amssbm}
\lim_{n\to \infty}\frac{1}{n+1}\log \bar{f}(Y_0^n)=a, \;\; \bar{P}-a.s.
\end{equation}

\noindent {\bf Proof of~(\ref{finiteness}).} To show~(\ref{finiteness}), it suffices to show that
$$H(Y_0^n|Y_{n+1}^{n+k(n)})<\infty\ \mbox{and}\ H(Y_0^n|Y_{n+1}^{\infty})>-\infty.$$
Using the fact that conditioning reduces entropy, we have
$$H(Y_0^n|Y_{n+1}^{n+k(n)})\le H(Y_0^n)\stackrel{(a)}{\le}  \frac{n+1}{2}\log 2\pi e M_3,$$
where $(a)$ follows from Corollary~\ref{entropybd}. Similarly,
\begin{align*}
H(Y_0^n)\ge H(Y_0^n|Y_{n+1}^{\infty})&\ge H(Y_0^n|X_0^{\infty}, Y_{n+1}^{\infty})\\
&\stackrel{(a)}{=}H(Y_0^n|X_0^{n+1}, Y_{n+1})\\
&=H(Y_0^{n+1}|X_0^{n+1})-H(Y_{n+1}|X_{0}^{n+1})\\
&\stackrel{(b)}{>}-\infty,
\end{align*}
where $(a)$ follows from the fact that $Y_0^n$ is independent of $(X_{n+2}^{\infty},Y_{n+2}^{\infty})$ given $(X_0^{n+1},Y_{n+1})$ and $(b)$ follows from Corollary~(\ref{entropybd}).

\noindent {\bf Proof of~(\ref{amssbm}).}
Let
$$
H_{\bar{P}}(Y_0^n) \triangleq \EE_{\bar{P}}[-\log \bar{f}(Y_0^n)].
$$
To establish (\ref{amssbm}), we will apply the generalized Shannon-McMillan-Breiman theorem~(Theorem $1$ in \cite{barron}), for which we need to verify that $\{Y_i\}$ under the probability measure $\bar{P}$ is stationary and ergodic and $\abs{H_{\bar{P}}(Y_{n}|Y_0^{n-1})}<\infty$.

From Lemma~\ref{amsproperty} it follows that $\{Y_i\}$ under the probability measure $\bar{P}$ is stationary and ergodic.  From Lemma~\ref{amsexistence} it follows that $\EE_{\bar{P}}[Y_i^2]\le M_3$.
Then
\begin{equation}\label{n+1}
H_{\bar{P}}(Y_0^n)\le \sum_{i=0}^{n}H_{\bar{P}}(Y_i)\stackrel{(a)}{\le} \frac{n+1}{2}\log 2\pi eM_3<\infty,
\end{equation}
where $(a)$ follows from the the fact that Gaussian distribution maximizes entropy for a given variance.

Since $f(Y_{k}^{k+n})\le 1$, by (\ref{im31}), we have
\begin{equation}\label{positive} H_{\bar{P}}(Y_0^{n})=\EE_{\bar{P}}[-\log \bar{f}(Y_{0}^{n})]\ge 0.\end{equation}
Combining~(\ref{n+1}) and~(\ref{positive}), we deduce
$$
\abs{H_{\bar{P}}(Y_{n}|Y_0^{n-1})}\le H_{\bar{P}}(Y_0^{n})+H_{\bar{P}}(Y_0^{n-1})<\infty,
$$
as desired.
\end{proof}

We are now ready to prove the asymptotic equipartition property for $\{Y_n\}$ and $\{X_n, Y_n\}$.
\begin{theorem} \label{aep}
The following two limits exist
$$
H(Y) \triangleq \lim_{n\to \infty}\frac{1}{n+1}H(Y_0^n), \quad H(X, Y) \triangleq \lim_{n\to \infty}\frac{1}{n+1}H(X_0^n;Y_0^n),
$$
and therefore,
$$
I(X;Y)=\lim_{n\to \infty}\frac{1}{n+1}I(X_0^n;Y_0^n)
$$
also exists. Moreover,
$$
\lim_{n\to \infty}-\frac{1}{n+1}\log f_{Y_0^n}(Y_0^n)=H(Y), \;\; P-a.s.
$$
and
$$
\lim_{n\to \infty}-\frac{1}{n+1}\log f_{X_0^n, Y_0^n}(X_0^n, Y_0^n)=H(X,Y), \;\; P-a.s.
$$
\end{theorem}

\begin{proof}

We only show the existence of $H(Y)$, the proof of that of $H(X, Y)$ being completely parallel. Apparently, the existence of $H(Y)$ and $H(X, Y)$ immediately implies that of
$I(X; Y)$.

By Lemma~\ref{amsexistence}, we have $\EE_{\bar{P}}[Y_n^2]<\infty$. Then it follows from the Birkhoff's ergodic theorem~\cite{durrettprobability} that
$$
\lim_{n\to \infty}\frac{1}{n+1}\sum_{i=0}^{n}Y_i^2 \mbox{ exists}, \;\; \bar{P}-a.s.
$$
and
$$
\EE_{\bar{P}}\left[\lim_{n\to \infty}\frac{1}{n+1}\sum_{i=0}^{n}Y_i^2\right]=\EE_{\bar{P}}[Y_1^2].
$$
From Theorems~\ref{ergodicthm} and~\ref{bothergodic}, it follows that
$$
\lim_{n\to \infty}\frac{1}{n+1}\sum_{i=0}^{n}Y_i^2 \mbox{ exists}, \;\; P-a.s.
$$
and
{$$\EE_P \left[\lim_{n\to \infty}\frac{1}{n+1}\sum_{i=0}^{n}Y_i^2\right]=\EE_{\bar{P}}[Y_1^2].$$}
And from Lemma~\ref{amsexistence}, it follows that
{
$$
\EE_P \left[\lim_{n\to \infty}\frac{1}{n+1}\sum_{i=0}^{n}Y_i^2\right]=\lim_{n\to\infty}\frac{1}{n+1}\sum_{i=0}^{n}\EE_P [Y_i^2]=\EE_{\bar{P}}[Y_1^2].
$$}
As shown in Lemma~\ref{existenceentropy}, we have
{
$$
\lim_{n\to \infty}\frac{1}{n+1}\log f_{Y_0^n}(Y_0^n)=a, \;\; P-a.s.
$$}
It then follows from (\ref{densitybd}) and the general dominated convergence theorem~\cite{roydenrealanalysis} that
\begin{align*}
\EE_P\left[\lim_{n\to\infty}\frac{1}{n+1}\log f(Y_0^n)\right]&=\lim_{n\to\infty}\EE_P\left[\frac{1}{n+1}\log f(Y_0^n)\right]\\
&=\lim_{n\to\infty}\frac{H(Y_0^n)}{n+1}\\
&=H(Y),
\end{align*}
which implies that $a=H(Y)$ and thereby yields the desired convergence.
\end{proof}

\section{Main Results} \label{stationarycapacity}
The {\em stationary capacity} $C_S$ and the $m$-th order {\em Markov capacity} $C_{Markov}^{(m)}$ of our channel are defined as
$$
C_{S} = \sup_{X}I(X;Y)\quad\mbox{and}\quad C_{Markov}^{(m)} = \sup_{X} I(X;Y),
$$
where the first supremum is taken over all the stationary and ergodic processes and the second one is over all the $m$-th order stationary and ergodic Markov chains. Now we are ready to state our main theorem, which relates various defined capacities above.
\begin{theorem}\label{main}
$$
C=C_{Shannon}=C_S=\lim_{m \to \infty} C_{Markov}^{(m)}.
$$
\end{theorem}

Our theorem confirms that for the channel~(\ref{sm-2}), the operational capacity can be approached by the Markov capacity, which justifies the effectiveness of the Markov approximation scheme in terms of computing the operational capacity.
\begin{proof} To prove the theorem, it suffices to prove that
$$
C_S\le C\le C_{Shannon}\le C_S=\lim_{m \to \infty} C_{Markov}^{(m)}.
$$

{\bf\underline{Proof of $C_S\le C$.}} This follows from a usual ``achievability part'' proof: For any rate $R<C_S$ and $\eps>0$, choose a stationary ergodic input process ${X_n}$ such that $R<I(X;Y)-\eps$. As shown in Theorem~\ref{aep}, $\{X_n,Y_n\}$ satisfies the AEP, we can complete the proof of the achievability by going through the usual random coding argument.

{\bf \underline{Proof of $C\le C_{Shannon}$.}} This follows from a usual ``converse part'' proof.

{\bf \underline{Proof of $C_{Shannon}\le C_S$.}} The proof is similar to the one in~\cite{Feinstein}, so we just outline the main steps.

{\bf Step 0.} First of all, for any $\eps>0$, choose $l$ such that
\begin{equation} \label{choose-l}
\sigma_B^{2l}\le \frac{\eps}{2(\sigma_A^2M_0^{2}+2\sigma_B^2(M_3+\sigma_E^2))\log M},
\end{equation}
and then $N$ and $X_0^N \sim p(x_0^{N})$ such that
\begin{equation} \label{choose-N}
\frac{l}{N}\le\frac{\eps}{4\log M}, \quad \frac{1}{N+1}I(X_0^N;Y_0^N)\ge C_{Shannon}-\eps.
\end{equation}

{\bf Step 1.} Now, let $\{\hat{X}_n\}$ be the ``independent block'' process defined as follows:
\begin{itemize}
\item[(i)] $(\hat{X}_{k(N+1)},\cdots,\hat{X}_{(k+1)(N+1)-1})$ are i.i.d. for $k=0,1,\cdots$;
\item[(ii)] $(\hat{X}_{0},\cdots,\hat{X}_{N})$ has the same distribution as $(X_0,X_1,\cdots, X_{N})$.
\end{itemize}
And let $\hat{Y}$ be the output obtained by passing $\hat{X}$ through the channel~(\ref{sm-2}). Let $\nu$ be independent of $\{\hat{X}_n\}$ and uniformly distributed over $\{0,1\cdots,N\}$, and let $\bar{X}_n=\hat{X}_{\nu+n}$. It can be verified that $\{\bar{X}_n\}$ is a stationary and ergodic process.

{\bf Step 2.} Let $\{\bar{Y}_n\}$ be the output obtained by passing the stationary process $\{\bar{X}_n\}$ through the channel~(\ref{sm-2}). Letting
$$
I(\bar{X};\bar{Y})=\lim_{n\to \infty}\frac{1}{n+1}I(\bar{X}_0^n;\bar{Y}_0^{n}),
$$
we will show that
\begin{equation}
\label{sdp}I(\bar{X};\bar{Y})-\frac{1}{N+1}I(X_0^{N};Y_0^{N})\ge -\eps,\end{equation}
which, by the arbitrariness of $\eps$, will imply the claim.

Note that it can be verified that
\begin{align*}
{p}_{\bar{X}_0^n}(x_0^n)f_{\bar{Y}_0^n|\bar{X}_0^n}(y_0^n|x_0^n)&=\sum_{k=0}^{N}\frac{1}{N+1}{P}(\bar{X}_0^n=x_0^n|\nu=k)f_{\bar{Y}_0^n|\bar{X}_0^n}(y_0^n|x_0^n)\\
 &=\sum_{k=0}^{N}\frac{1}{N+1}{P}(\hat{X}_k^{k+n}=x_0^n|\nu=k)f_{\bar{Y}_0^n|\bar{X}_0^n}(y_0^n|x_0^n)\\
 &=\sum_{k=0}^{N}\frac{1}{N+1}{P}(\hat{X}_{(k),0}^{n}=x_0^n|\nu=k)f_{\bar{Y}_0^n|\bar{X}_0^n}(y_0^n|x_0^n),
 \end{align*}
where $\hat{X}_{(k),n} \triangleq \hat{X}_{k+n}$. For $k=0, 2, \cdots, N$, let $\hat{Y}_{(k)}=\{\hat{Y}_{(k),n}\}$ denote the output process obtained by passing the process $\hat{X}_{(k)}=\{\hat{X}_{(k),n}\}$ through the channel (\ref{sm-2}). Then it follows from Lemma 2 in~\cite{Feinstein} that
$$
 I(\bar{X};\bar{Y})=\frac{1}{N+1}\sum_{j=0}^{N}I(\hat{X}_{(k)};\hat{Y}_{(k)}),
$$
where
$$
I(\hat{X}_{(k)};\hat{Y}_{(k)})=\lim_{n\to\infty}\frac{1}{n+1}I(\hat{X}_{(k),0}^{n};\hat{Y}_{(k),0}^{n}).
$$

To prove~(\ref{sdp}), it suffices to establish that for any $k$,
\begin{align}\label{im15}
I(\hat{X}_{(k)};\hat{Y}_{(k)})&\ge \frac{1}{N+1}I(X_0^N;Y_0^N)-\eps.
\end{align}
The proof of~({\ref{im15}}) for a general $k$ are similar, so in the following we only show it holds true for $k=0$. Here, we note that when $k=0$,
$$
I(\hat{X}_{(k)};\hat{Y}_{(k)})=I(\hat{X};\hat{Y})=\lim_{l \to \infty} \frac{1}{l (N+1)} I(\hat{X}_0^{l(N+1)-1};\hat{Y}_0^{l(N+1)-1}).
$$
Using the chain rule for mutual information, we have
\begin{align*}
I(\hat{X}_0^{l(N+1)-1};\hat{Y}_0^{l(N+1)-1})\ge\sum_{i=1}^{l} I(\hat{X}_{(i-1)(N+1)}^{i(N+1)-1};\hat{Y}_{(i-1)(N+1)}^{i(N+1)-1}|\hat{X}_0^{(i-1)(N+1)-1},\hat{Y}_0^{(i-1)(N+1)-1}),
\end{align*}
which means that, to prove~(\ref{im15}), it suffices to show that
\begin{align*}
\frac{1}{N+1}I(\hat{X}_{(i-1)(N+1)}^{i(N+1)-1};\hat{Y}_{(i-1)(N+1)}^{i(N+1)-1}|\hat{X}_0^{(i-1)(N+1)-1},\hat{Y}_0^{(i-1)(N+1)-1})&\ge\frac{1}{N+1}I(\hat{X}_0^{N+1};\hat{Y}_0^{N+1})-\eps.
\end{align*}
Without loss of generality, we prove this holds true for $i=2$.
Note that
\begin{align}
I(\hat{X}_{N+1}^{2N+1};\hat{Y}_{N+1}^{2N+1}|\hat{X}_0^{N},\hat{Y}_0^{N})
&= \sum_{x_0^{N}} p_{\hat{X}_{0}^N}(x_0^{N})\int f_{\hat{Y}_0^N|\hat{X}_0^N}(y_0^{N}|x_0^{N})I(\hat{X}_{N+1}^{2N+1};\hat{Y}_{N+1}^{2N+1}|x_0^{N},y_0^{N}) {d}y_0^{N} \notag\\
&= \sum_{x_0^{N}}p_{\hat{X}_{0}^N}(x_0^{N})\int f_{\hat{Y}_0^N|\hat{X}_0^n}(y_0^{N}|x_0^{N}) I(\hat{X}_{N+1}^{2N+1};\hat{Y}_{N+1}^{2N+1}|x_{N},y_{N}) {d}y_0^{N}\notag\\
&= \sum_{x_{N}}p_{\hat{X}_{0}^N}(x_0^{N})\int  f_{\hat{Y}_{N}|\hat{X}_{N}}(y_N|x_N)I(\hat{X}_{N+1}^{2N+1};\hat{Y}_{N+1}^{2N+1}|x_{N},y_{N}) {d}y_{N}. \notag
\end{align}
It follows from Statement $a)$ in Proposition~\ref{pr} that
\begin{align*}
&\hspace{-1cm} |I(\hat{X}_{N+1}^{2N+1};\hat{Y}_{N+1}^{2N+1}|x_{N},y_{N})-I(\hat{X}_0^{N+1};\hat{Y}_0^{N+1})|\\
&\le 2(l+1)\log M+(N-l)\log M\times( \sigma_A^2x_{N}^{2}+2\sigma_B^2({y}_{N}^{2}+\sigma_E^2))\sigma_B^{2l},
\end{align*}
which implies that
\begin{align}
I(\hat{X}_{N+1}^{2N+1};\hat{Y}_{N+1}^{2N+1}|\hat{X}_0^{N},\hat{Y}_0^{N})
{}&\ge \sum_{x_{N}}p_{X_N}(x_{N})\int f_{Y_N|X_N}(y_{N}|x_{N})\notag\\
{}&\hspace{-4cm}\times \left\{I(\hat{X}_0^{N+1};\hat{Y}_0^{N+1})-2(l+1)\log M-\left[(N-l)\log M \times (\sigma_A^2x_{N}^{2}+2\sigma_B^2({y}_{N}^{2}+\sigma_E^2))\sigma_B^{2l}\right]\right\} {d}y_{N} \notag\\
{}&\hspace{-4cm} \stackrel{(a)}{\ge} I(\hat{X}_0^{N+1};\hat{Y}_0^{N+1})-2(l+1)\log M-\left\{(N-l)\log M\right.\times\left.(\sigma_A^2M_0^{2}+2\sigma_B^2(M_3+\sigma_E^2))\sigma_B^{2l}\right\}, \notag
\end{align}
where $(a)$ follows from Corollary~\ref{uniformintegrable}. Now, with (\ref{choose-l}) and (\ref{choose-N}), we conclude that 
\begin{align*}
\frac{1}{N+1}I(\hat{X}_{N+1}^{2N+1};\hat{Y}_{N+1}^{2N+1}|\hat{X}_0^{N},\hat{Y}_0^{N})
&\ge \frac{1}{N+1}I(\hat{X}_0^{N+1};\hat{Y}_0^{N+1})-\eps\\
&=\frac{1}{N+1}I({X}_0^{N+1};{Y}_0^{N+1})-\eps,
\end{align*}
as desired.

\underline{\bf {Proof of $C_S=\lim_{m\to\infty}C_{Markov}^{(m)}$.}} To prove this, we only need to show that for any $\eps > 0$, one can find an $m$-th order stationary and ergodic Markov chain $\tilde{X}$ such that
$$
I(\tilde{X}; \tilde{Y}) \geq C_S -\eps,
$$
where $\tilde{Y}$ is the output process obtained when passing $\tilde{X}$ through the channel (\ref{sm-2}).

First of all, let $X$ be a stationary process such that
$$
I(X;Y)\ge C_S-\eps/3.
$$
Now, construct the $m$-th order stationary and ergodic Markov chain $\tilde{X}$ by setting
$$
P(\tilde{X}_0^m=x_0^m)=P(X_0^m=x_0^m),
$$
and let $\tilde{Y}$ be the output processes obtained by passing $\tilde{X}$ through the channel~(\ref{sm-2}).

It follows from Statement $b)$ in Proposition~\ref{pr} that for any $m, i \geq 0$,
\begin{eqnarray*}
&&\hspace{-2cm} \frac{1}{m+1}I(\tilde{X}_{im+i}^{(i+1)m+1};\tilde{Y}_{im+i}^{(i+1)m+1})- \frac{1}{m+1}I(\tilde{X}_{0}^{m};\tilde{Y}_{0}^{m})\\
&\ge&-\frac{3(k+1)\log 2\pi e M_3}{n+1}-2M_{3}\pi e(M_8+3M_9)\sigma_{B}^{2k}\\
&{}&-\frac{1}{m+1}(M_4+M_5M_3)-\left(2M_3M_6+\frac{4M_1M_3M_7}{(1-\sigma_B)^2}+\frac{12M_3M_7}{(m+1)(1-\sigma_B^2)}\right)\sigma_{B}^{2k}.
\end{eqnarray*}
Choosing $m$ and $k$ sufficiently large, we have
$$
\frac{1}{m+1}I(\tilde{X}_{im+i}^{(i+1)m+1};\tilde{Y}_{im+i}^{(i+1)m+1})\ge\frac{1}{m+1}I(\tilde{X}_{0}^{m};\tilde{Y}_{0}^{m})-\frac{\eps}{3},
$$
which, together with the chain rule for entropy and the fact that $H(\tilde{X}_{im+i}^{(i+1)m+i})=H(\tilde{X}_{0}^{m})$, implies that
\begin{align}\label{uppcon}
H(\tilde{X}|\tilde{Y})&=\lim_{s\to \infty}\frac{1}{s(m+1)}H(\tilde{X}_{0}^{s(m+1)-1}|\tilde{Y}_{0}^{s(m+1)-1})\nonumber\\
&=\lim_{s\to \infty}\frac{1}{s(m+1)}\sum_{i=0}^{s-1}H(\tilde{X}_{im+i}^{(i+1)m+i}|\tilde{X}_{0}^{im+i-1},\tilde{Y}_{0}^{s(m+1)-1})\nonumber\\
&\le \lim_{s\to \infty}\frac{1}{s(m+1)}\sum_{i=0}^{s-1}H(\tilde{X}_{im+i}^{(i+1)m+i}|\tilde{Y}_{im+i}^{(i+1)m+i})\nonumber\\
&= \lim_{s\to \infty}\frac{1}{s(m+1)}\sum_{i=0}^{s-1}\left\{H(\tilde{X}_{im+i}^{(i+1)m+i})-I(\tilde{X}_{im+i}^{(i+1)m+i};\tilde{Y}_{im+i}^{(i+1)m+i})\right\}\nonumber\\
&\le \lim_{s\to \infty}\frac{1}{s(m+1)}\sum_{i=0}^{s-1}\left\{H(\tilde{X}_{0}^{m})-I(\tilde{X}_{0}^{m};\tilde{Y}_{0}^{m})+\eps\right\}\nonumber\\
&=\frac{1}{m+1}H(\tilde{X}_{0}^{m}|\tilde{Y}_0^m)+\frac{\eps}{m+1}\nonumber\\
&= \frac{1}{m+1}H(X_{0}^{m}|Y_0^m)+\frac{\eps}{m+1}.
\end{align}
Now, choosing $m$ sufficiently large such that
$$
\frac{1}{m+1}H(X_{0}^{m}|Y_0^m)\le H(X|Y)+\eps/3,
$$
and using (\ref{uppcon}) and the stationary property of $X$, we deduce that
\begin{align*}
I(\tilde{X};\tilde{Y})&=H(\tilde{X})-H(\tilde{X}|\tilde{Y})\\
&\ge H(\tilde{X})- \frac{1}{m+1}H(X_{0}^{m}|Y_0^m)\\
&= H(\tilde{X}_{m}|\tilde{X}_{0}^{m-1})- \frac{1}{m+1}H(X_{0}^{m}|Y_0^m)-\frac{\eps}{m+1}\\
&= H(X_{m}|{X}_{0}^{m-1})- \frac{1}{m+1}H(X_{0}^{m}|Y_0^m)-\frac{\eps}{m+1}\\
&\ge H(X)- \frac{1}{m+1}H(X_{0}^{m}|Y_0^m)-\frac{\eps}{m+1}\\
&\ge H(X)-H(X|Y)-\frac{\eps}{3}-\frac{\eps}{m+1}\\
&\ge I(X;Y)-2\eps/3\\
&\ge C_S-\eps,
\end{align*}
as desired.
\end{proof}

\section{Conclusion and Future Work}

In this paper, via an information-theoretic analysis, we prove that, for a recently proposed one dimensional causal flash memory channel~\cite{kavcic2014}, as the order tends to infinity, its Markov capacity converges to its operational capacity, which translates to the theoretical limit of memory cell storage efficiency.

The aforementioned result serves as a first step to the journey of investigating whether the ideas and techniques in the theory of finite-state channels can be instrumental to compute the capacity of flash memory channels. A natural follow-up question in the future is the concavity of the mutual information rate of flash memory channels with respect to the parameters of an input Markov process, which is a much desired property that will help ensure the convergence of the capacity computing algorithms in~\cite{vontobel, randomapproachhan}. Here, we note that the concavity of the mutual information rate has been established for special classes of finite-state channels~\cite{hm09b,lihan2013,lihan2014}.

Further investigations are needed to be conducted to see whether the ideas and techniques developed in this work can be applied/adapted to the two dimensional model in~\cite{kavcic2014}, a more realistic channel model for flash memories. Our preliminary investigations indicate that despite some technical issues such as anti-causality (which naturally arises in a two dimensional channel), the framework laid out in this work, coupled with a possible conversion from two dimensional models to one dimensional models via appropriate re-indexing, will likely encompass an effective approach to two dimensional flash memory channels.

\section*{Appendices} \appendix

\section{Proof of~(\ref{squaredistance})}\label{Appendix-A}
The proof follows from a similar argument in~\cite{madras2010}.
Without loss of generality, we assume $0<\sigma_1<\sigma_2$ and let $\phi(x;\sigma,\mu)=\frac{1}{\sqrt{2\pi \sigma_1^2}}e^{-\frac{(x-\mu)^2}{2\sigma{2}}}$. Then
\begin{eqnarray*}
&&\hspace{-1.5cm}\int_{-\infty}^{\infty}x^2 \left|\frac{1}{\sqrt{2\pi \sigma_1^2}}e^{-\frac{(x-\mu)^2}{2\sigma_1^{2}}}- \frac{1}{\sqrt{2\pi \sigma_2^2}}e^{-\frac{(x-\mu)^2}{2\sigma_2^{2}}} \right| \,{d} x \\
&=&\left\{\int_{\{x:\phi(x;\sigma_1,\mu)>\phi(x;\sigma_2,\mu) \}}+\int_{\{x:\phi(x;\sigma_1,\mu)<\phi(x;\sigma_2,\mu) \}}\right\}x^2 \left|\phi(x;\sigma_1,\mu)-\phi(x;\sigma_2,\mu) \right| \,{d} x \\
&=&2\int_{\{x:\phi(x;\sigma_1,\mu)>\phi(x;\sigma_2,\mu) \}} \left(\phi(x;\sigma_1,\mu)-\phi(x;\sigma_2,\mu) \right) \,{d} x +\int_{-\infty}^{\infty} x^2(\phi(x;\sigma_2,\mu)-\phi(x;\sigma_1,\mu))dx\\
&\le&\sigma_2^2-\sigma_1^2+2\frac{\sigma_2-\sigma_1}{\sigma_1\sigma_2}\int_{-\infty}^{\infty}  \frac{x^2}{\sqrt{2\pi}}e^{-\frac{(x-\mu)^2}{2\sigma_1^{2}}}dx\\
&\le& \sigma_2^2-\sigma_1^2+2\frac{\sigma_1^3(\sigma_2^2-\sigma_1^2)}{(\sigma_1+\sigma_2)\sigma_1\sigma_2}\\
&\le&3|\sigma_{1}^{2}-\sigma_{2}^{2}|.
\end{eqnarray*}

\section{Proofs of (\ref{newdensitybd}) and (\ref{densitybd})} \label{Appendix-B}

We first conduct some preparatory computations before the proofs.

Note that given $E_1^n=e_1^n$, $U_0^n=u_0^n$, $X_0^n=x_0^n$ and $Y_{i-1}=y_{i-1}$, $Y_i$ is a Gaussian random variable with density
$$
f(y_i|y_{i-1},x_{i-1}^i,e_i,u_i)=\frac{1}{\sqrt{2\pi(\sigma_A^2 x_{i-1}^2+\sigma_B^2(y_{i-1}-e_i)^2+1)}}e^{-\frac{(y_i-x_i-u_i)^2}{2(\sigma_A^2 x_{i-1}^2+\sigma_B^2(y_{i-1}-e_i)^2+1)}}.
$$
Clearly, $f(y_i|y_{i-1},x_{i-1}^i,e_i,u_i,y_{i-1})\le 1$ and for $i\ge 1$,
\begin{align}\label{lowerbound}
f_{Y_i|X_{i-1}^i,Y_{i-1}}f(y_i|y_{i-1},x_{i-1}^i)&=\int de_i du_i f_{U_i}(u_i)f_{E_i}(e_i) f(y_i|y_{i-1},x_{i-1}^i,e_i,u_i,y_{i-1})\notag\\
&= \int de_i du_i \frac{f_{U_i}(u_i)f_{E_i}(e_i)}{\sqrt{2\pi(\sigma_A^2 x_{i-1}^2+\sigma_B^2(y_{i-1}-e_i)^2+1)}}e^{-\frac{(y_i-x_i-u_i)^2}{2}}\nonumber\\
&\ge\int de_i du_i \frac{f_{U_i}(u_i)f_{E_i}(e_i)}{\sqrt{2\pi(\sigma_A^2 x_{i-1}^2+\sigma_B^2(y_{i-1}-e_i)^2+1)}}e^{-\frac{3(y_i^2+x_i^2+u_i^2)}{2}}\notag\\
&\ge\int de_i du_i \frac{f_{U_i}(u_i)f_{E_i}(e_i)}{\sqrt{2\pi(\sigma_A^2 M_0^2+2\sigma_B^2(y_{i-1}^2+e_i^2)+1)}}e^{-\frac{3(y_i^2+2M_0^2)}{2}}\notag\\
&\ge \int_{-1}^{1} f_{E_i}(e_i)de_i  \frac{1}{\sqrt{2\pi(\sigma_A^2 M_0^2+2\sigma_B^2(y_{i-1}^2+1)+1)}}e^{-\frac{3(y_i^2+2M_0^2)}{2}},\notag
\end{align}
where $M_0$ is as in (\ref{M_0}).

{\bf Proof of~(\ref{newdensitybd})} For any $\tilde{M} >0$, we have
\begin{align}
&\hspace{-1cm} f_{Y_{m+k-1}|X_{m+k-1}^{m+n}}(y|x_{k-1}^n) \notag\\
&=\sum_{\tilde{x}_{0}^{m+k-2}}\left\{p_{X_{0}^{m+k-2}|X_{m+k-1}^{m+n}}(\tilde{x}_{0}^{m+k-2}|x_{k-1}^n)\right. \notag\\
&\times \left.\int f_{Y_{m+k-2}|X_{0}^{m+k-2}}(\tilde{y}|\tilde{x}_{0}^{m+k-2})f_{Y_{m+k-1}|X_{m+k-2}^{m+k-1},Y_{m+k-2}}(y|\tilde{x}_{m+k-2},x_{k-1},\tilde{y})\right\} d\tilde{y}\notag\\
&\ge\sum_{\tilde{x}_{0}^{m+k-2}}\left\{p_{X_{0}^{m+k-2}|X_{m+k-1}^{m+n}}(\tilde{x}_{0}^{m+k-2}|x_{k-1}^n)\right.\notag\\
&\times \left.\int f_{Y_{m+k-2}|X_{0}^{m+k-2}}(\tilde{y}|\tilde{x}_{0}^{m+k-2})\frac{\int_{-1}^{1} f_{E_i}(e_i)de_i  }{\sqrt{2\pi(\sigma_A^2 M_0^2+2\sigma_B^2(\tilde{y}^2+1)+1)}}e^{-\frac{3(y^2+2M_0^2)}{2}}\right\} d\tilde{y} \notag\\
&\ge\sum_{\tilde{x}_{0}^{m+k-2}}\left\{p_{X_{0}^{m+k-2}|X_{m+k-1}^{m+n}}(\tilde{x}_{0}^{m+k-2}|x_{k-1}^n)\right.\notag\\
&\times \left.\int_{-\tilde{M}}^{\tilde{M}} f_{Y_{m+k-2}|X_{0}^{m+k-2}}(\tilde{y}|\tilde{x}_{0}^{m+k-2})\frac{\int_{-1}^{1} f_{E_i}(e_i)de_i  }{\sqrt{2\pi(\sigma_A^2 M_0^2+2\sigma_B^2(\tilde{y}^2+1)+1)}}e^{-\frac{3(y^2+2M_0^2)}{2}}\right\} d\tilde{y} \notag\\
&\ge\sum_{\tilde{x}_{0}^{m+k-2}}\left\{p_{X_{0}^{m+k-2}|X_{m+k-1}^{m+n}}(\tilde{x}_{0}^{M+k-2}|x_{k-1}^n)P(|Y_{m+k-2}|\le \tilde{M}|X_{0}^{m+k-2}=\tilde{x}_0^{m+k-2})\right.\notag\\
&\times \left. \frac{\int_{-1}^{1} f_{E_i}(e_i)de_i  }{\sqrt{2\pi(\sigma_A^2 M_0^2+2\sigma_B^2(\tilde{M}^2+1)+1)}}e^{-\frac{3(y^2+2M_0^2)}{2}}\right\}. \notag
 \end{align}
It then follows from Corollary~(\ref{uniformintegrable}) and the Markov inequality that
$$
P(|Y_{m+k-2}|\le \tilde{M}|X_{0}^{m+k-2}=\tilde{x}_0^{m+k-2})\ge 1-\frac{\EE[Y_{m+k-2}^2|X_{0}^{m+k-2}=\tilde{x}_0^{m+k-2}]}{\tilde{M}^2}\ge 1-\frac{M_3}{\tilde{M}^2}.
$$
If $\tilde{M}$ is chosen such that for all $\tilde{x}_0^{m+k-2}$
$$
P(|Y_{m+k-2}|\le \tilde{M}|X_{0}^{m+k-2}=\tilde{x}_0^{m+k-2}) \ge 1/2,
$$
we then have
$$
\log f_{Y_{m+k-1}|X_{m+k-1}^{m+n}}(y|x_{k-1}^n) \ge \log \frac{\int_{-1}^{1} f_{E_i}(e_i)de_i  }{2\sqrt{2\pi(\sigma_A^2 M_0^2+2\sigma_B^2(\tilde{M}^2+1)+1)}}e^{-3M_0^2}-3y^2.
$$
The desired result then follows by choosing
$$
M_4=\abs{\log \frac{\int_{-1}^{1} f_{E_i}(e_i)de_i  }{2\sqrt{2\pi(\sigma_A^2 M_0^2+2\sigma_B^2(\tilde{M}^2+1)+1)}}e^{-3M_0^2}}\quad\mbox{and}\quad M_5=3.
$$

{\bf Proof of~(\ref{densitybd}).}
Note that
\begin{align}
f(y_m^{m+n})&=\sum_{x_0^{m+n}}p(x_0^{m+n})\int f(y_{m-1}|x_0^{m-1})f(y_{m}^{m+n}|x_{0}^{m+n},y_{m-1}) dy_{m-1} \notag\\
&=\sum_{x_0^{m+n}}p(x_0^{m+n})\int f(y_{m-1}|x_0^{m-1}) \prod_{i=m}^{m+n}f(y_i|y_{i-1},x_{i-1}^i) dy_{m-1} \nonumber\\
&=\sum_{x_0^{m+n}}p(x_0^{m+n})\int f(y_{m-1}|x_0^{m-1}) \prod_{i=m}^{m+n}f(y_i|y_{i-1},x_{i-1}^i) dy_{m-1} \notag\\
&=\sum_{x_0^{m+n}}p(x_0^{m+n})\int f(y_{m-1}|x_0^{m-1}) \prod_{i=m}^{m+n}\int f(e_i)f(u_i)f(y_i|y_{i-1},x_{i-1}^i,u_i,e_i)dy_{m-1}.\notag
\end{align}
It then follows that $f(y_{m}^{m+n})\le 1$ and a similar argument as in the proof of (\ref{newdensitybd}) that for sufficiently large $\tilde{M}$,
\begin{align}
f(y_m^{m+n})&\ge \sum_{x_0^{m+n}}p(x_0^{m+n})\int  f(y_{m-1}|x_0^{m-1})\prod_{i=m}^{m+n}  \frac{\int_{-1}^{1} f_{E_i}(e_i)de_i }{\sqrt{2\pi(\sigma_A^2 M_0^2+2\sigma_B^2(y_{i-1}^2+1)+1)}}e^{-\frac{3(y_i^2+2M_0^2)}{2}} dy_{m-1} \nonumber\\
&\ge \left(\sum_{x_0^{m-1}}p(x_0^{m-1})\int f(y_{m-1}|x_0^{m-1})\frac{\int_{-1}^{1} f_{E_m}(e_m)de_m }{\sqrt{2\pi(\sigma_A^2 M_0^2+2\sigma_B^2(y_{m-1}^2+1)+1)}}e^{-\frac{3(y_m^2+2M_0^2)}{2}}\right) dy_{m-1} \notag\\
&\hspace{4mm}\times \left(\prod_{i=m+1}^{m+n}  \frac{\int_{-1}^{1} f_{E_i}(e_i)de_i }{\sqrt{2\pi(\sigma_A^2 M_0^2+2\sigma_B^2(y_{i-1}^2+1)+1)}}e^{-\frac{3(y_i^2+2M_0^2)}{2}}\right)\nonumber\\
&\ge \left(\sum_{x_0^{m-1}}p(x_0^{m-1})\int_{-\tilde{M}}^{\tilde{M}} f(y_{m-1}|x_0^{m-1})\frac{\int_{-1}^{1} f_{E_m}(e_m)de_m }{\sqrt{2\pi(\sigma_A^2 M_0^2+2\sigma_B^2(y_{m-1}^2+1)+1)}}e^{-\frac{3(y_m^2+2M_0^2)}{2}}\right) dy_{m-1} \notag\\
&\hspace{4mm}\times \left(\prod_{i=m+1}^{m+n}  \frac{\int_{-1}^{1} f_{E_i}(e_i)de_i }{\sqrt{2\pi(\sigma_A^2 M_0^2+2\sigma_B^2(y_{i-1}^2+1)+1)}}e^{-\frac{3(y_i^2+2M_0^2)}{2}}\right)\nonumber\\
&\ge \left(\frac{\int_{-1}^{1} f_{E_m}(e_m)de_m }{\sqrt{8\pi(\sigma_A^2 M_0^2+2\sigma_B^2(\tilde{M}^2+1)+1)}}e^{-\frac{3(y_m^2+2M_0^2)}{2}}\right)\notag\\
&\hspace{4mm}\times \left(\prod_{i=m+1}^{m+n}  \frac{\int_{-1}^{1} f_{E_i}(e_i)de_i }{\sqrt{2\pi(\sigma_A^2 M_0^2+2\sigma_B^2(y_{i-1}^2+1)+1)}}e^{-\frac{3(y_i^2+2M_0^2)}{2}}\right)\nonumber.
\end{align}

Now, we have
\begin{align}
0&\ge \log f(Y_m^{m+n})\notag\\
&\ge \log \frac{\int_{-1}^{1} f_{E_m}(e_m)de_m }{\sqrt{8\pi(\sigma_A^2 M_0^2+2\sigma_B^2(\tilde{M}^2+1)+1)}}e^{-\frac{3(y_m^2+2M_0^2)}{2}}\notag\\
&\hspace{4mm}+\log\left\{\prod_{i=m+1}^{m+n}  \frac{\int_{-1}^{1} f_{E_i}(e_i)de_i }{\sqrt{2\pi(\sigma_A^2 M_0^2+2\sigma_B^2(y_{i-1}^2+1)+1)}}e^{-\frac{3(y_i^2+2M_0^2)}{2}}\right\}\notag\\
&=-\sum_{i=m}^{m+n}\frac{3Y_i^2+6M_0^2}{2}+(n+1)\log\left(\frac{ \int_{-1}^{1} f(e_1)\,{d}e_1}{\sqrt{2\pi}}\right)-\frac{\log 2(\sigma_A^2 M_0^2+2\sigma_B^2(\tilde{M}^2+1)+1)}{2}\notag\\
&\hspace{4mm}-\frac{1}{2}\sum_{i=m+1}^{m+n}\log (1+\sigma_A^2M_0^2+2\sigma_B^2(Y_{i-1}^2+1))\notag\\
&\stackrel{(a)}{\ge} -\sum_{i=m}^{m+n}\frac{3Y_i^2+6M_0^2}{2}+(n+1)\log\left(\frac{ \int_{-1}^{1} f(e_1)\,{d}e_1}{\sqrt{2\pi}}\right)-\frac{\log 4(\sigma_A^2 M_0^2+2\sigma_B^2(\tilde{M}^2+1)+1)}{2}\notag\\
&\hspace{4mm}-\frac{1}{2}\sum_{i=m+1}^{m+n} (\sigma_A^2M_0^2+2\sigma_B^2(Y_{i-1}^2+1))\notag
\end{align}
where we have used the well-known inequality $\log(1+z)\le z$ for any $z > -1$ to derive $(a)$.

The desired (\ref{densitybd}) then follows by choosing
$$
M_6=\frac{(6+\sigma_A^2)M_0^2}{2}+\frac{2\sigma_{B}^{2}+\log 4(\sigma_A^2 M_0^2+2\sigma_B^2(\tilde{M}^2+1)+1)}{2}+\abs{\log\left(\frac{ \int_{-1}^{1} f(e_1)\,{d}e_1}{\sqrt{2\pi}}\right)}$$
and
$$
\quad M_7=\frac{3+2\sigma_{B}^2}{2}.
$$

\section{Proof of Lemma~\ref{amsexistence}} \label{Appendix-C}
For simplicity, we prove Lemma~\ref{amsexistence} for $n=1$, the proof for a general $n$ being similar.

Conditioned on $x_0^n, b_1^n$ and $u_0^n$, $Y_n$ is a Gaussian random variable with mean
$$\sum_{i=0}^{n}(x_i+u_i)\prod_{j=i+1}^{n}b_{j}$$
and variance
$$\prod_{j=1}^{n}b_{j}^2+\sum_{i=1}^{n}(x_{i-1}^2\sigma_A^2+b_i\sigma_E^2+1)\prod_{j=i+1}^{n}b_{j}^2\ge 1. $$
Let $\phi(y;\mu,\sigma^2)$ be the Gaussian density with mean $\mu$ and variance $\sigma^2$. Then the density of $Y_n$ is
$$f_{Y_n}(y)=\EE\left[\phi \left(y;\sum_{i=0}^{n}(X_i+U_i)\prod_{j=i+1}^{n}B_{j}^2,\prod_{j=1}^{n}B_j^2+\sum_{i=1}^{n}(X_{i-1}^2\sigma_A^2+B_i^2\sigma_E^2+1)\prod_{j=i+1}^{n}B_{j}^2\right)\right].$$
Since the processes $\{X_n\}$, $\{U_n\}$ and $\{B_n\}$ are all stationary, $f_{Y_n}(y)$ can be written as the following
{\small\begin{align*}
f_{Y_n}(y)=\EE\left[\phi \left(y;\sum_{i=-n}^{0}(X_i+U_i)\prod_{j=i+1}^{0}B_{j},\prod_{j=-n+1}^{0}B_j^2+\sum_{i=-n+1}^{0}(X_{i-1}^2\sigma_A^2+B_i^2\sigma_E^2+1)\prod_{j=i+1}^{0}B_{j}^2\right)\right].
\end{align*}}
Since
\begin{align*}
\sum_{i=-\infty}^{0}\EE\left[\left|(X_i+U_i)\prod_{j=i+1}^{0}B_{j}\right|\right]&\le 2M_0\sum_{i=-\infty}^{0}\prod_{j=i+1}^{0}(\EE B_{j}^2)^{\frac{1}{2}}\le 2M_0\sum_{i=-\infty}^{0}\sigma_B^{-i}<\infty,
\end{align*}
it follows from Theorem~3.1 in~\cite{rosalsky} that with probability $1$,
$$\sum_{i=-n}^{0}(X_i+U_i)\prod_{j=i+1}^{0}B_{j}$$ converges.

For any $\eps>0$,
\[\sum_{n=1}^{\infty}P\left(\prod_{j=-n+1}^{0}B_j^2>\eps\right)\le \sum_{n=1}^{\infty}\frac{\EE\left[\prod_{j=-n+1}^{0}B_j^2\right]}{\eps}= \sum_{n=1}^{\infty}\frac{\sigma_B^{2n}}{\eps}<\infty.\]
Then it follows from the Borel-Cantelli lemma that with probability $1$,
\begin{equation}\label{first} \prod_{j=-n+1}^{0}B_j^2\to 0.\end{equation}
Clearly, with probability $1$,
\begin{equation}\label{second}
\sum_{i=-n+1}^{0}(X_{i-1}^2\sigma_A^2+B_i\sigma_E^2+1)\prod_{j=i+1}^{0}B_{j}^2\to \sum_{i=-\infty}^{0}(X_{i-1}^2\sigma_A^2+B_i\sigma_E^2+1)\prod_{j=i+1}^{0}B_{j}^2.
\end{equation}
From~(\ref{first}) and~(\ref{second}), we have that with probability $1$,
\[
\prod_{j=-n+1}^{0}B_j^2+\sum_{i=-n+1}^{0}(X_{i-1}^2\sigma_A^2+B_i\sigma_E^2+1)\prod_{j=i+1}^{0}B_{j}^2\to \sum_{i=-\infty}^{0}(X_{i-1}^2\sigma_A^2+B_i\sigma_E^2+1)\prod_{j=i+1}^{0}B_{j}^2.
\]
Since
\begin{align*}
\EE\left[\sum_{i=-n+1}^{0}(X_{i-1}^2\sigma_A^2+B_i^2\sigma_E^2+1)\prod_{j=i+1}^{0}B_{j}^2\right]&\le \sum_{i=-n+1}^{-1}(M_0^2\sigma_A^2+\sigma_B^2\sigma_E^2+1)\sigma_B^{-2i}\\
&\le \sum_{i=-\infty}^{-1}(M_0^2\sigma_A^2+\sigma_B^2\sigma_E^2+1)\sigma_B^{-2i},
\end{align*}
it follows from Fatou's lemma~\cite{liptser} that
$$\EE\left[\sum_{i=-\infty}^{0}(X_{i-1}^2\sigma_A^2+B_i\sigma_E^2+1)\prod_{j=i+1}^{0}B_{j}^2\right]\le\sum_{i=-\infty}^{-1}(M_0^2\sigma_A^2+\sigma_B^2\sigma_E^2+1)\sigma_B^{2i},$$
which further implies that, with probability $1$,
$$\prod_{j=-\infty}^{0}B_j^2\sum_{i=-\infty}^{0}(X_{i-1}^2\sigma_A^2+B_i\sigma_E^2+1)\prod_{j=i+1}^{0}B_{j}^2<\infty.$$
It then follows from the bounded convergence theorem~\cite{liptser} that
\begin{align*}
f_{Y_n}(y)\to \EE\left[\phi \left(y;\sum_{i=-\infty}^{0}(X_i+U_i)\prod_{j=i+1}^{0}B_{j},\prod_{j=-\infty}^{0}B_j^2+\sum_{i=-\infty}^{0}(X_{i-1}^2\sigma_A^2+B_i^2\sigma_E^2+1)\prod_{j=i+1}^{0}B_{j}^2\right)\right].
\end{align*}
Let
\[
g(y)=\EE\left[\phi \left(y;\sum_{i=-\infty}^{0}(X_i+U_i)\prod_{j=i+1}^{0}B_{j},\prod_{j=-\infty}^{0}B_j^2+\sum_{i=-\infty}^{0}(X_{i-1}^2\sigma_A^2+B_i^2\sigma_E^2+1)\prod_{j=i+1}^{0}B_{j}^2\right)\right].
\]
Then for any Borel set $A\in\mathbb{R}$,
\[
\int_A \bar{f}_{Y_0}(y)\,{d}y\stackrel{(a)}{=}\bar{P}(Y_1\in A)=\lim_{n\to \infty}P(Y_n\in A)=\lim_{n\to \infty}\int_A f_{Y_n}(y)\,{d}y\stackrel{(b)}{=}\int_A\lim_{n\to \infty} f_{Y_n}(y)\,{d}y=\int_A g(y)\,{d}y,
\]
where $(a)$ follows from Theorem~\ref{amsproperty} and $(b)$ follows from $f_{Y_n}(y)\le 1$ and the bounded dominated convergence theorem~\cite{liptser}.
Therefore, $\bar{f}_{Y_0}(y)=g(y)=\lim_{n\to \infty}f_{Y_n}(y),$ which implies that $P(Y_n=\cdot)$ converges weakly to $\bar{P}(Y_0=\cdot)$. As shown in Corollary~\ref{uniformintegrable}, $\{Y_n^2\}$ under the probability measure $P$ is uniformly integrable. Then from Theorem 3.3 in~\cite{billingsley}, it follows that
$$
\EE_{\bar{P}}[Y_n^2]=\lim_{n\to \infty}\EE[Y_n^2]\le M_3<\infty.
$$

\end{document}